\documentclass[12pt,a4paper]{article}
\usepackage[margin=2.5cm]{geometry}
\usepackage[english]{babel}
\usepackage[T1]{fontenc}
\usepackage{latexsym,amsmath,amsfonts,amssymb,amsthm,hyperref,doi,tikz,float,verbatim}
\usetikzlibrary{arrows}
\newcommand{\arXiv}[2]{arXiv:\href{http://arxiv.org/abs/#1}{#1 #2}}

\hypersetup{%
pdfauthor={Tam\'{a}s~F.~G\"{o}rbe and Martin A. Halln\"{a}s},%
pdftitle={Quantization and explicit diagonalization of new compactified trigonometric Ruijsenaars-Schneider systems},%
pdfsubject={Preprint (\today)},%
pdfkeywords={Calogero-Moser-Sutherland,Ruijsenaars-Schneider,quantisation,Macdonald polynomials}}

\def\cA{\mathcal{A}}
\def\cD{\mathcal{D}}
\def\cE{\mathcal{E}}
\def\cH{\mathcal{H}}
\def\cN{\mathcal{N}}
\def\cS{\mathcal{S}}
\def\C{\mathbb{C}}
\def\CP{\mathbb{CP}}
\def\N{\mathbb{N}}
\def\R{\mathbb{R}}
\def\Z{\mathbb{Z}}
\DeclareMathOperator{\sgn}{sgn}
\newcommand{\vect}[1]{\boldsymbol{#1}}

\theoremstyle{plain}
\newtheorem{lemma}{Lemma}
\newtheorem{proposition}{Proposition}
\newtheorem{theorem}{Theorem}
\newtheorem{corollary}{Corollary}

\begin{document}

\title{Quantization and explicit diagonalization of new compactified trigonometric Ruijsenaars-Schneider systems}

\author{
Tam\'{a}s F. G\"{o}rbe\footnote{tfgorbe@physx.u-szeged.hu}\\
Department of Theoretical Physics\\
University of Szeged\\
Tisza Lajos krt 84-86, H-6720 Szeged, Hungary\\
and\\
Martin A. Halln\"{a}s\footnote{hallnas@chalmers.se}\\
Department of Mathematical Sciences\\
Chalmers University of Technology and the University of Gothenburg\\
SE-412 96 Gothenburg, Sweden}

\date{\today}

\maketitle

\begin{abstract}\noindent
Recently, Feh\'{e}r and Kluck discovered, at the level of classical mechanics, new compactified trigonometric Ruijsenaars-Schneider $n$-particle systems, with phase space symplectomorphic to the $(n-1)$-dimensional complex projective space. In this article, we quantize the so-called type (i) instances of these systems and explicitly solve the joint eigenvalue problem for the corresponding quantum Hamiltonians by generalising previous results of van Diejen and Vinet. Specifically, the quantum Hamiltonians are realized as discrete difference operators acting in a finite-dimensional Hilbert space of complex-valued functions supported on a uniform lattice over the classical configuration space, and their joint eigenfunctions are constructed in terms of discretized $A_{n-1}$ Macdonald polynomials with unitary parameters.
\end{abstract}

\newpage

\section{Introduction}\label{sec:intro}
The many-body systems introduced by Ruijsenaars and Schneider \cite{RS} are remarkable examples of integrable and exactly solvable one-dimensional $n$-particle models, which contain systems of Calogero-Moser-Sutherland and Toda type as limiting cases and have strong ties with numerous areas of both mathematics and physics. See, e.g., \cite{vDV00,Ne99,Rui90,Rui99} for reviews. Although Ruijsenaars-Schneider (RS) systems have now been studied for over three decades, they continue to attract attention as new applications and new variants are being found. A case in point is Feh\'{e}r and Kluck's recent discovery \cite{FKl} of novel instances of so-called compactified trigonometric RS system at the level of classical mechanics. In this paper, we introduce the quantum versions of these models (given by any so-called type (i) value of the coupling parameter that satisfies a natural quantisation condition) and explicitly solve the corresponding joint eigenvalue problems by generalising earlier results of van Diejen and Vinet \cite{vDV}.

To explain our results in more detail, we recall that the standard trigonometric RS system is defined by the Hamiltonian 
$$
\tilde{H}(g;\vect{x},\vect{p})=\sum_{j=1}^n\cosh(\beta p_j)
\sqrt{\displaystyle\prod_{\substack{k=1\\k\neq j}}^n
\bigg(1+\frac{\sinh^2\big(\frac{\alpha\beta g}{2}\big)}
{\sin^2\frac{\alpha}{2}(x_j-x_k)}\bigg)},
$$
with generalised coordinates $\vect{x}=(x_1,\dots,x_n)$, conjugate momenta $\vect{p}=(p_1,\dots,p_n)$ and real-valued scale parameters $\alpha,\beta$ and coupling parameter $g$. Its phase space can be chosen to be the unbounded domain $\{(\vect{x},\vect{p})\in\R^{2n}\mid x_1>\dots>x_n\}$. Applying a Wick rotation, i.e., replacing the real parameter $\beta$ with the imaginary $\mathrm{i}\beta$ in $\tilde{H}$, leads to the Hamiltonian we are interested in, namely the Hamiltonian of the \emph{compactified} trigonometric RS system,
\begin{equation}\label{H}
H(g;\vect{x},\vect{p})=\sum_{j=1}^n\cos(\beta p_j)
\sqrt{\displaystyle\prod_{\substack{k=1\\k\neq j}}^n
\bigg(1-\frac{\sin^2\big(\frac{\alpha\beta g}{2}\big)}
{\sin^2\frac{\alpha}{2}(x_j-x_k)}\bigg)},
\end{equation}
which was first introduced and studied by Ruijsenaars \cite{Rui90,Rui}. Note that the hyperbolic functions turned into their trigonometric counterparts and there is a sign change under the square root. These important differences entail a Hamiltonian that is periodic not only in the coordinates $\vect{x}$ but also in the momenta $\vect{p}$, and the unbounded phase space we had previously is naturally replaced by a bounded phase space, which, in fact, becomes compact after a suitable completion.

From here onwards, we assume, without loss of generality, that $\alpha>0$ and work with
$$
\beta=1,
$$
but the parameter $\beta$ is easily reintroduced by substituting $\alpha\to\alpha\beta$, $x_j\to\beta^{-1}x_j$ and $p_j\to\beta p_j$. Note that the Hamiltonian $H$ (with $\beta=1$) is $2\pi/\alpha$-periodic in $g$ and $g=0$ corresponds to a system of free particles. Therefore, it suffices to consider couplings satisfying\footnote{\label{footnote:1}The range of couplings yielding genuinely different dynamics is $(0,\pi/\alpha)$, since $H$ is invariant with respect to the replacement $g\to 2\pi/\alpha-g$. However, we find it more convenient to work with \eqref{g}.}
\begin{equation}\label{g}
0<g<\frac{2\pi}{\alpha}.
\end{equation}

The $n$-particle compactified trigonometric RS system is Liouville integrable \cite{Rui87}, i.e., there exists a Poisson commuting\footnote{With respect to the canonical symplectic form $\sum_{j=1}^ndx_j\wedge dp_j$.} set of $n$ independent functions $H_1,\dots,H_n$ to which $H$ \eqref{H} belongs. More specifically, one can take
\begin{equation}\label{Hr}
H_r(g;\vect{x},\vect{p})=\sum_{\substack{J\subset\{1,\dots,n\}\\|J|=r}}
s(g;J)\cos(\textstyle\sum_{j\in J}p_j)
\sqrt{\displaystyle\prod_{\substack{j\in J\\k\notin J}}
\bigg(1-\frac{\sin^2\big(\frac{\alpha g}{2}\big)}{\sin^2\frac{\alpha}{2}(x_j-x_k)}\bigg)},\quad r=1,\dots,n,
\end{equation}
where the sign-factors $s(g;J)\in\{1,-1\}$ (to be defined later) satisfy $s(g;\{1\})=\dots=s(g;\{n\})$ and $s(g;\{1,\dots,n\})=1$ for all values of $g$. The first relation entails that $H_1$ \eqref{Hr} is equal to $H$ \eqref{H} (up to an overall sign), while the second one means that
$$
H_n(g;\vect{x},\vect{p})=\cos(p_1+\dots+p_n).
$$
This entails the conservation of the `total momentum' $p_1+\dots+p_n$. Since the center-of-mass $(x_1+\cdots+x_n)/n$ evolves linearly in time (see, e.g., \cite{RS}), it is thus natural to impose $p_1+\cdots+p_n=0$ and restrict attention to the relative motion in the center-of-mass hyperplane
$$
E_n=\{\vect{x}\in\R^n\mid x_1+\dots+x_n=0\},
$$
which is governed by reduced Hamiltonians $\cH_r$, $r=1,\dots,n-1$.

When first introducing the compactified trigonometric RS system, Ruijsenaars considered the configuration space given by
\begin{equation}\label{Rconfig}
x_j-x_{j+1}>g,\ \ j=1,\dots,n-1,\quad x_1-x_n<\frac{2\pi}{\alpha}-g.
\end{equation}
This ensures that all factors under the square roots in \eqref{H} and \eqref{Hr} are positive, so that the reduced Hamiltonians $\cH_r(\vect{x},\vect{p})$ are real-valued and smooth, but clearly requires
\begin{equation}\label{standard-g}
0<g<\frac{2\pi}{\alpha n}.
\end{equation}
As will become clear in Subsection \ref{sec:reduced}, in this case all signs $s(g;J)$ in \eqref{Hr} are equal to one. Due to the periodicity of $\cH_r(\vect{x},\vect{p})$ in both $\vect{x}$ and $\vect{p}$, it is natural to employ a bounded phase space, which, as Ruijsenaars \cite{Rui} showed, can be embedded densely and symplectically in the $(n-1)$-dimensional complex projective space equipped with a multiple of the standard Fubini-Study symplectic form.

By allowing generic $g$-values \eqref{g}, Feh\'{e}r and Kluck \cite{FKl} recently obtained new compactified trigonometric Ruijsenaars-Schneider systems. As the coupling parameter $g$ varies, the pertinent (local) configuration space, which always contains the equal-distance configuration
$$
x_j-x_{j+1}=2\pi/\alpha n,\quad j=1,\ldots,n-1,
$$
can change between two drastically different forms. The corresponding parameter values, referred to as type (i) and (ii), form disjoint open subintervals that (except for a finite number of excluded points) partition the interval $(0,2\pi/\alpha)$, see Fig. \ref{Fig:gintervals}. Fixing the number of particles $n\geq 2$, the type (i) couplings can be described as follows. Choose a parameter $p\in\{1,\dots,n-1\}$ such that $p$ is coprime to $n$, i.e.~$\gcd(n,p)=1$, and let $q$ be the multiplicative inverse of $p$ in the ring $\Z_n$, that is $pq\equiv 1\pmod{n}$. Then $g$ can be picked from a certain interval punctured at $2\pi p/\alpha n$, namely $g$ has to satisfy either
\begin{equation}\label{typei}
\frac{2\pi}{\alpha}\bigg(\frac{p}{n}-\frac{1}{nq}\bigg)<
g<\frac{2\pi}{\alpha}\frac{p}{n}\quad\text{or}\quad
\frac{2\pi}{\alpha}\frac{p}{n}<
g<\frac{2\pi}{\alpha}\bigg(\frac{p}{n}+\frac{1}{n(n-q)}\bigg).
\end{equation}
Note that setting $p=q=1$ in the first formula of \eqref{typei} produces the interval \eqref{standard-g}. New models with type (i) couplings arise for $n\geq 3$, while new type (ii) models first appear at $n=4$, see Fig. \ref{Fig:gintervals}. In the type (i) cases, the closure of the configuration space is always contained within the Weyl alcove
\begin{equation}\label{An}
\cA_n=\{\vect{x}\in E_n\mid x_1>\dots>x_n>x_1-2\pi/\alpha\}
\end{equation}
and the particles cannot collide. This is no longer the case for type (ii) couplings.

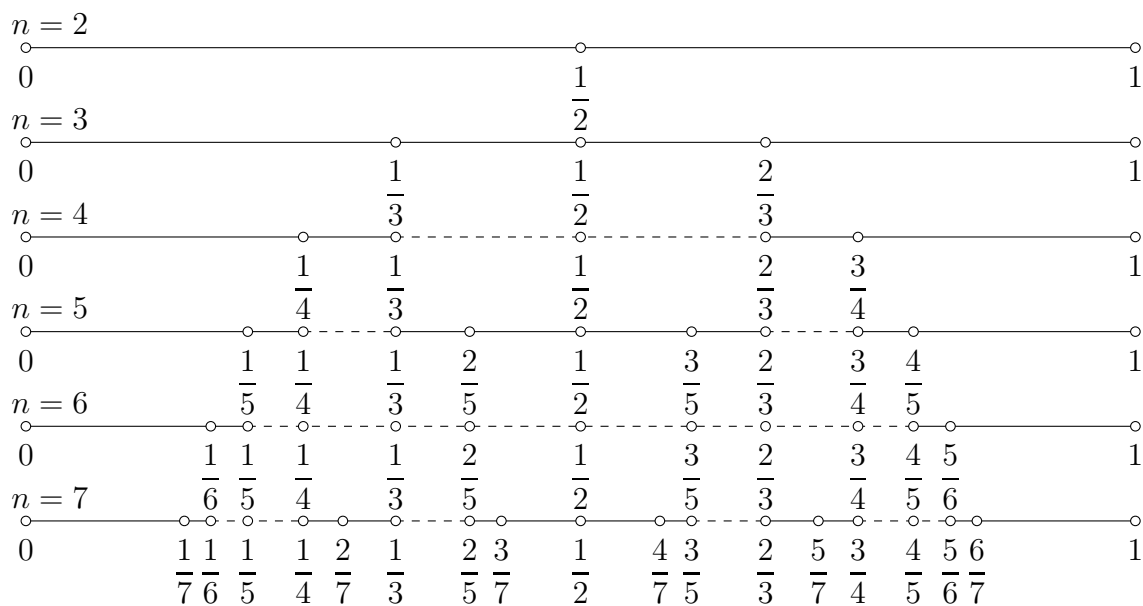
\begin{figure}[H]
\centering
\begin{tikzpicture}[scale=.76]
\def\s{1.2\textwidth}
\def\r{.2em}
\draw(1em,9em) node{$n=2$};
\draw (0,8em)--({\s},8em);
\draw[black,fill=white]
(0,8em) circle(\r) node[below,yshift=-1mm]{$0$}
(\s/2,8em) circle(\r) node[below,yshift=-1mm]{$\displaystyle\frac{1}{2}$}
(\s,8em) circle(\r) node[below,yshift=-1mm]{$1$};
\draw(1em,5em) node{$n=3$};
\draw (0,4em)--({\s},4em);
\draw[black,fill=white]
(0,4em) circle(\r) node[below,yshift=-1mm]{$0$}
(\s/3,4em) circle(\r) node[below,yshift=-1mm]{$\displaystyle\frac{1}{3}$}
(\s/2,4em) circle(\r) node[below,yshift=-1mm]{$\displaystyle\frac{1}{2}$}
(2*\s/3,4em) circle(\r) node[below,yshift=-1mm]{$\displaystyle\frac{2}{3}$}
(\s,4em) circle(\r) node[below,yshift=-1mm]{$1$};
\draw(1em,1em) node{$n=4$};
\draw (0,0)--({\s/3},0) ({2*\s/3},0)--({\s},0);
\draw[dashed] ({\s/3},0)--({2*\s/3},0);
\draw[black,fill=white]
(0,0) circle(\r) node[below,yshift=-1mm]{$0$}
(\s/4,0) circle(\r) node[below,yshift=-1mm]{$\displaystyle\frac{1}{4}$}
(\s/3,0) circle(\r) node[below,yshift=-1mm]{$\displaystyle\frac{1}{3}$}
(\s/2,0) circle(\r) node[below,yshift=-1mm]{$\displaystyle\frac{1}{2}$}
(2*\s/3,0) circle(\r) node[below,yshift=-1mm]{$\displaystyle\frac{2}{3}$}
(3*\s/4,0) circle(\r) node[below,yshift=-1mm]{$\displaystyle\frac{3}{4}$}
(\s,0) circle(\r) node[below,yshift=-1mm]{$1$};
\draw(1em,-3em) node{$n=5$};
\draw (0,-4em)--(\s/4,-4em) (\s/3,-4em)--(2*\s/3,-4em) (3*\s/4,-4em)--(\s,-4em);
\draw[dashed] (\s/4,-4em)--(\s/3,-4em) (2*\s/3,-4em)--(3*\s/4,-4em);
\draw[black,fill=white]
(0,-4em) circle(\r) node[below,yshift=-1mm]{$0$}
(\s/5,-4em) circle(\r) node[below,yshift=-1mm]{$\displaystyle\frac{1}{5}$}
(\s/4,-4em) circle(\r) node[below,yshift=-1mm]{$\displaystyle\frac{1}{4}$}
(\s/3,-4em) circle(\r) node[below,yshift=-1mm]{$\displaystyle\frac{1}{3}$}
(2*\s/5,-4em) circle(\r) node[below,yshift=-1mm]{$\displaystyle\frac{2}{5}$}
(\s/2,-4em) circle(\r) node[below,yshift=-1mm]{$\displaystyle\frac{1}{2}$}
(3*\s/5,-4em) circle(\r) node[below,yshift=-1mm]{$\displaystyle\frac{3}{5}$}
(2*\s/3,-4em) circle(\r) node[below,yshift=-1mm]{$\displaystyle\frac{2}{3}$}
(3*\s/4,-4em) circle(\r) node[below,yshift=-1mm]{$\displaystyle\frac{3}{4}$}
(4*\s/5,-4em) circle(\r) node[below,yshift=-1mm]{$\displaystyle\frac{4}{5}$}
(\s,-4em) circle(\r) node[below,yshift=-1mm]{$1$};
\draw(1em,-7em) node{$n=6$};
\draw (0,-8em)--(\s/5,-8em) (4*\s/5,-8em)--(\s,-8em);
\draw[dashed] (\s/5,-8em)--(4*\s/5,-8em);
\draw[black,fill=white]
(0,-8em) circle(\r) node[below,yshift=-1mm]{$0$}
(\s/6,-8em) circle(\r) node[below,yshift=-1mm]{$\displaystyle\frac{1}{6}$}
(\s/5,-8em) circle(\r) node[below,yshift=-1mm]{$\displaystyle\frac{1}{5}$}
(\s/4,-8em) circle(\r) node[below,yshift=-1mm]{$\displaystyle\frac{1}{4}$}
(\s/3,-8em) circle(\r) node[below,yshift=-1mm]{$\displaystyle\frac{1}{3}$}
(2*\s/5,-8em) circle(\r) node[below,yshift=-1mm]{$\displaystyle\frac{2}{5}$}
(\s/2,-8em) circle(\r) node[below,yshift=-1mm]{$\displaystyle\frac{1}{2}$}
(3*\s/5,-8em) circle(\r) node[below,yshift=-1mm]{$\displaystyle\frac{3}{5}$}
(2*\s/3,-8em) circle(\r) node[below,yshift=-1mm]{$\displaystyle\frac{2}{3}$}
(3*\s/4,-8em) circle(\r) node[below,yshift=-1mm]{$\displaystyle\frac{3}{4}$}
(4*\s/5,-8em) circle(\r) node[below,yshift=-1mm]{$\displaystyle\frac{4}{5}$}
(5*\s/6,-8em) circle(\r) node[below,yshift=-1mm]{$\displaystyle\frac{5}{6}$}
(\s,-8em) circle(\r) node[below,yshift=-1mm]{$1$};
\draw(1em,-11em) node{$n=7$};
\draw (0,-12em)--(\s/6,-12em) (\s/4,-12em)--(\s/3,-12em) (2*\s/5,-12em)--(3*\s/5,-12em)
(2*\s/3,-12em)--(3*\s/4,-12em) (5*\s/6,-12em)--(\s,-12em);
\draw[dashed] (\s/6,-12em)--(\s/4,-12em) (\s/3,-12em)--(2*\s/5,-12em)
(3*\s/5,-12em)--(2*\s/3,-12em) (3*\s/4,-12em)--(5*\s/6,-12em);
\draw[black,fill=white]
(0,-12em) circle(\r) node[below,yshift=-1mm]{$0$}
(\s/7,-12em) circle(\r) node[below,yshift=-1mm]{$\displaystyle\frac{1}{7}$}
(\s/6,-12em) circle(\r) node[below,yshift=-1mm]{$\displaystyle\frac{1}{6}$}
(\s/5,-12em) circle(\r) node[below,yshift=-1mm]{$\displaystyle\frac{1}{5}$}
(\s/4,-12em) circle(\r) node[below,yshift=-1mm]{$\displaystyle\frac{1}{4}$}
(2*\s/7,-12em) circle(\r) node[below,yshift=-1mm]{$\displaystyle\frac{2}{7}$}
(\s/3,-12em) circle(\r) node[below,yshift=-1mm]{$\displaystyle\frac{1}{3}$}
(2*\s/5,-12em) circle(\r) node[below,yshift=-1mm]{$\displaystyle\frac{2}{5}$}
(3*\s/7,-12em) circle(\r) node[below,yshift=-1mm]{$\displaystyle\frac{3}{7}$}
(\s/2,-12em) circle(\r) node[below,yshift=-1mm]{$\displaystyle\frac{1}{2}$}
(4*\s/7,-12em) circle(\r) node[below,yshift=-1mm]{$\displaystyle\frac{4}{7}$}
(3*\s/5,-12em) circle(\r) node[below,yshift=-1mm]{$\displaystyle\frac{3}{5}$}
(2*\s/3,-12em) circle(\r) node[below,yshift=-1mm]{$\displaystyle\frac{2}{3}$}
(5*\s/7,-12em) circle(\r) node[below,yshift=-1mm]{$\displaystyle\frac{5}{7}$}
(3*\s/4,-12em) circle(\r) node[below,yshift=-1mm]{$\displaystyle\frac{3}{4}$}
(4*\s/5,-12em) circle(\r) node[below,yshift=-1mm]{$\displaystyle\frac{4}{5}$}
(5*\s/6,-12em) circle(\r) node[below,yshift=-1mm]{$\displaystyle\frac{5}{6}$}
(6*\s/7,-12em) circle(\r) node[below,yshift=-1mm]{$\displaystyle\frac{6}{7}$}
(\s,-12em) circle(\r) node[below,yshift=-1mm]{$1$};
\end{tikzpicture}
\caption{The range of $\alpha g/2\pi$ for $n=2,\dots,7$. The displayed values are excluded.
Admissible values of $g$ form intervals of type (i) (solid) and type (ii) (dashed) couplings. The leftmost intervals consist of standard $g$-values \eqref{standard-g}; the rest give rise to new models (up to $g\to 2\pi/\alpha-g$, cf. Footnote \ref{footnote:1}).}
\label{Fig:gintervals}
\end{figure}

Under the $g$-constraint \eqref{standard-g}, van Diejen and Vinet \cite{vDV} quantised the compactified trigonometric RS system. They realised the quantum analogues of the reduced classical Hamiltonians $\cH_r$, $r=1,\ldots,n-1$, as pair-wise commuting discrete difference operators acting in a finite-dimensional Hilbert space consisting of complex-valued functions supported on a uniform lattice over the classical configuration space \eqref{Rconfig}. Requiring the lattice to fit precisely over the configuration space, which guaranteed well-defined and self-adjoint quantum Hamiltonians, led to the quantisation condition
$$
\frac{2\pi}{\alpha}-ng\in\mathbb{N}.
$$
Assuming this condition to be satisfied, they obtained, in particular, the explicit solution of the corresponding joint eigenvalue problem in terms of discretised $A_{n-1}$ Macdonald polynomials with unitary parameters $q=e^{\mathrm{i}\alpha}$ and $t=q^g$. In the two-particle ($n=2$) case, the pertinent eigenvalue problem was solved earlier by Ruijsenaars \cite{Rui90}. Let us also note that Feh\'{e}r and Klim\v{c}\'{i}k \cite{FK} obtained the spectra of the quantum Hamiltonians in an alternative way, namely by means of K\"{a}hler quantization.

Our motivation for this paper stems from the following observation made at the classical level: RS systems with type (i) couplings \eqref{typei} have similar features as the ones with standard couplings \eqref{standard-g}. In particular, the phase space of systems with type (i) couplings is always symplectomorphic to the $(n-1)$-dimensional complex projective space $\CP^{n-1}$, equipped with the symplectic structure $|M|\omega_{\mathrm{FS}}$, where $\omega_{\mathrm{FS}}$ stands for the Fubini-Study form and $M$ is given by
\begin{equation}\label{M}
M=\frac{2\pi}{\alpha}p-ng.
\end{equation}
For further details see \cite{FG,FKl,Rui}. Moreover, the configuration space in question can be viewed as the interior $\Sigma_{g,p}$ of the $(n-1)$-dimensional simplex
\begin{equation}\label{bSigma}
\overline{\Sigma}_{g,p}\equiv \{\vect{x}\in E_n\mid \sgn(M)(x_j-x_{j+p}-g)\geq 0,\, j=1,\dots,n\}\subset\cA_n,
\end{equation}
where
$$
x_{n+k}\equiv x_k-\frac{2\pi}{\alpha}.
$$
An illustration of the difference between the cases $M>0$ and $M<0$ is plotted in Fig. \ref{Fig:config}. As we demonstrate in Section \ref{sec:quantisation}, the products under the square roots in \eqref{Hr} are non-negative in $\overline{\Sigma}_{g,p}$ and can only vanish on the boundary, which ensures that the reduced Hamiltonians $\cH_r$ are real-valued and smooth in $\Sigma_{g,p}$.

A possible interpretation of the model is that of $n$ particles, moving on a circle of radius one half, positioned at $X_j=\frac{1}{2}e^{\mathrm{i}\alpha x_j}$, $j=1,\dots,n$, with a pair-wise interaction that depends on the square of the chord-distance. The inequalities that define $\overline{\Sigma}_{g,p}$ \eqref{bSigma} then can be thought of as lower ($M>0$) or upper ($M<0$) bounds on the $p$-nearest neighbour distances between particles. This visual interpretation gives us an intuitive understanding of the importance of $p$ being coprime to $n$. Namely, only for $p$-parameters that are coprime to $n$ do the inequalities in \eqref{bSigma} control the distance between every pair of particles. See Fig. \ref{Fig:graph}. In addition, it also becomes clear why the configuration space is invariant under the following change in parameters $(g,p)\to(2\pi/\alpha-g,n-p)$, i.e.
$$
\overline{\Sigma}_{g,p}=\overline{\Sigma}_{\frac{2\pi}{\alpha}-g,n-p},\quad 0<g<\frac{2\pi}{\alpha},\ \ p\in\{1,\ldots,n-1\}.
$$

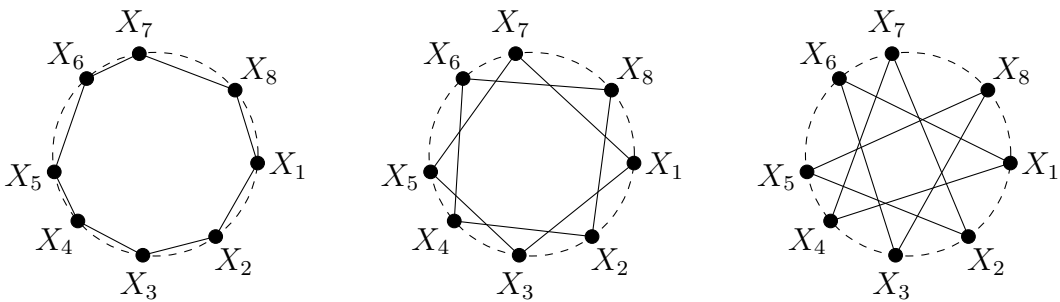
\begin{figure}[H]
\centering
\begin{tikzpicture}[scale=.9]
\def\R{1.5}
\def\r{.1}
\def\Angles{{355,306,263,221,190,132,99,39,15}}
\draw[dashed](0,0) circle(\R);
\foreach \x in {1,...,8}{
\draw[fill=black]({\R*cos(\Angles[Mod(\x-1,8)])},{\R*sin(\Angles[Mod(\x-1,8)])})circle(\r);
\draw({1.3*\R*cos(\Angles[Mod(\x-1,8)])},{1.3*\R*sin(\Angles[Mod(\x-1,8)])})node{$X_{\x}$};
\draw({\R*cos(\Angles[Mod(\x-1,8)])},{\R*sin(\Angles[Mod(\x-1,8)])})--({\R*cos(\Angles[Mod(\x,8)])},{\R*sin(\Angles[Mod(\x,8)])});}
\begin{scope}[shift={{(5.5,0)}}]
\draw[dashed](0,0) circle(\R);
\foreach \x in {1,...,8}{
\draw[fill=black]({\R*cos(\Angles[Mod(\x-1,8)])},{\R*sin(\Angles[Mod(\x-1,8)])})circle(\r);
\draw({1.3*\R*cos(\Angles[Mod(\x-1,8)])},{1.3*\R*sin(\Angles[Mod(\x-1,8)])})node{$X_{\x}$};
\draw({\R*cos(\Angles[Mod(\x-1,8)])},{\R*sin(\Angles[Mod(\x-1,8)])})--({\R*cos(\Angles[Mod(\x+1,8)])},{\R*sin(\Angles[Mod(\x+1,8)])});}
\end{scope}
\begin{scope}[shift={{(11,0)}}]
\draw[dashed](0,0) circle(\R);
\foreach \x in {1,...,8}{
\draw[fill=black]({\R*cos(\Angles[Mod(\x-1,8)])},{\R*sin(\Angles[Mod(\x-1,8)])})circle(\r);
\draw({1.3*\R*cos(\Angles[Mod(\x-1,8)])},{1.3*\R*sin(\Angles[Mod(\x-1,8)])})node{$X_{\x}$};
\draw({\R*cos(\Angles[Mod(\x-1,8)])},{\R*sin(\Angles[Mod(\x-1,8)])})--({\R*cos(\Angles[Mod(\x+2,8)])},{\R*sin(\Angles[Mod(\x+2,8)])});}
\end{scope}
\end{tikzpicture}
\caption{Chords connecting $p$-nearest neighbours for $n=8$ with $p=1,2,3$, respectively.
If the resulting graph (solid lines) is disconnected (such as the one in the middle), nearest neighbour particles can get arbitrarily close to each other regardless of lower/upper bounds on the drawn chord-distances. The number of graph components is $\gcd(n,p)$.
}
\label{Fig:graph}
\end{figure}

The state of affairs sketched above suggests that it should be possible to generalise the approach and results of van Diejen and Vinet \cite{vDV} from the standard $g$-values \eqref{standard-g} to all type (i) values \eqref{typei}. In this paper, we demonstrate that this is indeed the case. More precisely, for all type (i) values of the coupling parameter $g$ satisfying the quantisation condition $M\equiv 2\pi p/\alpha-ng\in \mathbb{Z}\setminus\{0\}$, we associate to the reduced classical Hamiltonians $\cH_r$, $r=1,\ldots,n-1$, pair-wise commuting and self-adjoint discrete difference operators acting in a finite-dimensional Hilbert space of complex-valued functions on a uniform lattice over the simplex \eqref{bSigma}. In addition, using the $A_{n-1}$ Macdonald polynomials we explicitly solve the corresponding joint eigenvalue problem.

As will become clear in the main text, many of the qualitative features of our results are encoded in the choice of $p$ and the sign of $M$ \eqref{M}, which distinguishes between the two cases in \eqref{typei}. Moreover, the presence of negative $\vect{x}$-dependent factors in \eqref{Hr} entails that some of our key (technical) results require more intricate or alternative proofs compared with the corresponding results in \cite{vDV}.

The rest of this paper is structured as follows. Section \ref{sec:quantisation} is devoted to the quantisation of the compactified trigonometric RS system for all type (i) couplings: after fixing root system notation that will be used throughout the paper, we introduce the pertinent finite-dimensional Hilbert space of lattice functions and associate self-adjoint and pair-wise commuting discrete difference operators to the reduced classical Hamiltonians. In Section \ref{sec:eigfuncs}, the corresponding joint eigenvalue problem is solved explicitly, and Section \ref{sec:discussion} contains a discussion of open problems and possible directions for future research.

\medskip

\textit{Note.} For our purposes, we find it convenient to employ the convention
$$
\N_0\equiv\{0,1,2,\dots\},\qquad\N\equiv\{1,2,\dots\}.
$$

\section{Quantisation}\label{sec:quantisation}
In this section, we quantise the family of Poisson commuting reduced classical Hamiltonians $\cH_r$, $r=1,\ldots,n-1$, (given by \eqref{cHr} below) for all type (i) values \eqref{typei} of the coupling parameter $g$. The corresponding quantum Hamiltonians will be given by commuting discrete difference operators acting in a finite-dimensional Hilbert space of lattice functions.

\subsection{Root system notation}\label{subsec:rootnot}
We begin by specifying root system notation that will be used throughout the remainder of the paper.

Given a positive integer $n$, we let $\{\vect{e}_1,\dots,\vect{e}_n\}$ denote the standard
basis in $\R^n$ and let $\langle\cdot,\cdot\rangle$ be the usual inner product on $\R^n$, so that $\langle \vect{e}_j,\vect{e}_k\rangle=\delta_{jk}$.

In this paper, we focus on the root systems
$$
A_{n-1} = \{\vect{e}_j-\vect{e}_k\mid
j,k=1,\dots,n,\ j\neq k\}\subset E_n,\quad n\geq 2,
$$
but a brief discussion of potential generalisations to other root systems can be found in Section \ref{sec:discussion}. For $p\in\{1,\dots,n-1\}$ coprime to $n$, in notation $p\perp n$, we find it convenient to make use of a specific $p$-dependent base $\{\vect{a}_{1,p},\dots,\vect{a}_{n-1,p}\}$ of $A_{n-1}$, consisting of the simple roots
\begin{equation}\label{ajp}
\vect{a}_{j,p}=\vect{e}_j-\vect{e}_{j+p},\quad j=1,\dots,n-1,
\end{equation}
where we employed the periodicity convention
$$
\vect{e}_{n+j}\equiv \vect{e}_j.
$$
To see that $\{\vect{a}_{j,p}\}$ is in fact a base of $A_{n-1}$, we recall that the Weyl group $W_R$ of a root system $R$ acts transitively on bases and $W_{A_{n-1}}\cong S_n$. Let $\sigma\in S_n$ be the permutation defined by $\sigma(j)=jp\mod n$. The corresponding element in the Weyl group $W_{A_{n-1}}$ sends the standard base vectors $\vect{a}_j=\vect{e}_j-\vect{e}_{j+1}$ to $\vect{a}_{\sigma(j),p}$, $j=1,\dots,n-1$.
We also need the associated fundamental weights $\{\vect{\omega}_{1,p},\dots,\vect{\omega}_{n-1,p}\}\subset E_n$, characterised by the property
$$
\langle\vect{a}_{j,p},\vect{\omega}_{k,p}\rangle=\delta_{jk},\quad j,k=1,\dots,n-1.
$$
Within the $A_{n-1}$ root lattice $Q$ and weight lattice $\Lambda$, we get the (integral) cones
$$
Q_p^+=\mathrm{span}_{\N_0}\{\vect{a}_{1,p},\dots,\vect{a}_{n-1,p}\}
\quad\text{and}\quad
\Lambda_p^+=\mathrm{span}_{\N_0}\{\vect{\omega}_{1,p},\dots,\vect{\omega}_{n-1,p}\}.
$$
The former cone contains the set of positive roots
$$
A_{n-1,p}^+=A_{n-1}\cap Q^+_p,
$$
which consists of roots of the form $\vect{a}_{i_1,p}+\dots+\vect{a}_{i_t,p}$ with distinct indices $i_1,\dots,i_t$. The cone $\Lambda_p^+$ (of dominant weights) is partially ordered by the dominance order, defined for $\vect{\lambda},\vect{\mu}\in\Lambda_p^+$ by
\begin{equation}\label{dominance}
\vect{\mu}\preceq\vect{\lambda}
\quad\text{iff}\quad
\vect{\lambda}-\vect{\mu}\in Q_p^+.
\end{equation}
Letting
$$
\vect{a}_{n,p}=\vect{e}_n-\vect{e}_p,
$$
we note that the corresponding maximal root $\vect{a}_{\max,p}$ in $A_{n-1}$ is given by
\begin{equation}\label{amax}
\vect{a}_{\max,p}=\vect{a}_{1,p}+\dots+\vect{a}_{n-1,p}=-\vect{a}_{n,p}.
\end{equation}

Specialising to $p=1$, we recover the standard base of $A_{n-1}$. For convenience, we then drop the subscript $p$ and simply write
\begin{equation}\label{aj}
\vect{a}_j\equiv\vect{a}_{j,1}=\vect{e}_j-\vect{e}_{j+1},\quad j=1,\dots,n-1,
\end{equation}
\begin{equation}\label{omegak}
\vect{\omega}_k\equiv \vect{\omega}_{k,1}=\vect{e}_1+\dots+\vect{e}_k-\frac{k}{n}(\vect{e}_1+\dots+\vect{e}_n),\quad k=1,\dots,n-1,
\end{equation}
and
\begin{equation}\label{QLamA}
Q^+\equiv Q^+_1,\qquad \Lambda^+\equiv\Lambda^+_1,\qquad A_{n-1}^+\equiv A_{n-1,1}^+.
\end{equation}

For general $p$-values, it will often be natural to make use of the weighted sum of fundamental weights
$$
\vect{\rho}_p=g(\vect{\omega}_{1,p}+\dots+\vect{\omega}_{n-p,p})
+\bigg(g-\frac{2\pi}{\alpha}\bigg)(\vect{\omega}_{n-p+1,p}
+\dots+\vect{\omega}_{n-1,p}),
$$
which in the $p=1$ case specialises to the standard weighted half-sum of positive roots
\begin{equation}\label{rho}
\vect{\rho}\equiv\vect{\rho}_1=g\sum_{j=1}^{n-1}\vect{\omega}_j= \frac{g}{2}\sum_{\vect{a}\in A_{n-1}^+}\vect{a}.
\end{equation}
In particular, it allows us to establish a convenient characterisation of the simplex $\overline{\Sigma}_{g,p}$ \eqref{bSigma} in terms of the $p$-dependent simple roots $\vect{a}_{j,p}$ \eqref{ajp} and the maximal root $\vect{a}_{\max,p}$ \eqref{amax}, as illustrated in Fig. \ref{Fig:config}.

\begin{lemma}\label{Lemma:simplex}
The simplex $\overline{\Sigma}_{g,p}$ consists of all points $\vect{x}\in\cA_n$ \eqref{An} satisfying the inequalities
\begin{enumerate}
\item $\sgn(M)\langle\vect{a}_{j,p},\vect{x}-\vect{\rho}_p\rangle\geq 0$ for $j=1,\dots,n-1$,
\item $\sgn(M)\langle\vect{a}_{\max,p},\vect{x}-\vect{\rho}_p\rangle\leq |M|$.
\end{enumerate}
\end{lemma}

\begin{proof}
With the periodicity convention $x_{j+n}\equiv x_j-2\pi/\alpha$ in effect, we have
$$
\langle\vect{a}_{j,p},\vect{x}\rangle
=\begin{cases}
x_j-x_{j+p},&j=1,\dots,n-p\\[.5em]
x_j-x_{j+p}-\dfrac{2\pi}{\alpha},&j=n-p+1,\dots,n-1
\end{cases}
$$
and
$$
\langle\vect{a}_{\max,p},\vect{x}-\vect{\rho}_p\rangle
=-\langle\vect{a}_{n,p},\vect{x}\rangle-\langle\vect{a}_{\max,p},\vect{\rho}_p\rangle
=-x_n+x_{n+p}+M+g.
$$
Hence the former set of inequalities is equivalent to $\sgn(M)(x_j-x_{j+p}-g)\geq 0$, $j=1,\dots,n-1$, and the last inequality can be rewritten as $\sgn(M)(x_n-x_{n+p}-g)\geq 0$.
\end{proof}

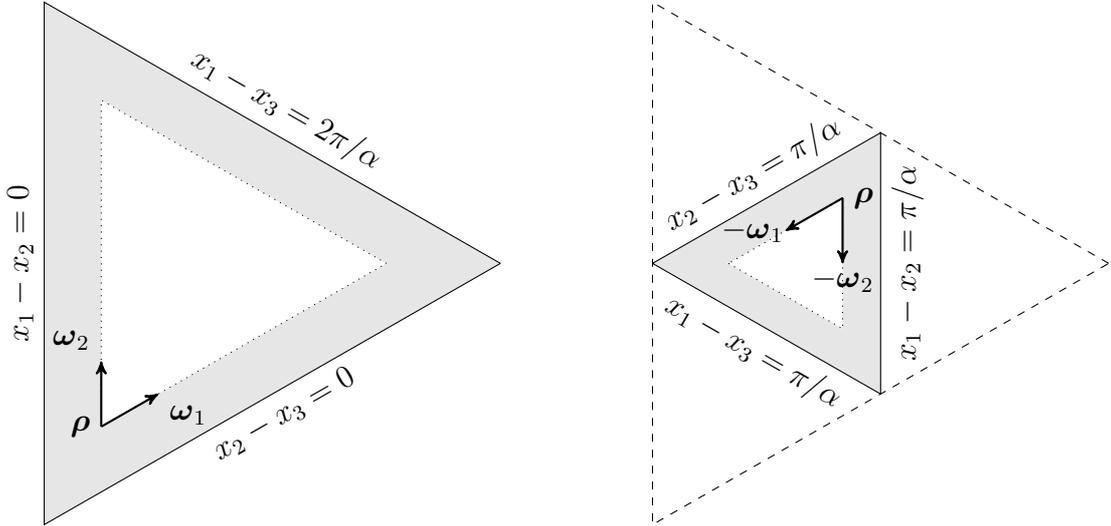
\begin{figure}[H]
\centering
\begin{tikzpicture}[black]
\def\Radius{4}
\def\radius{2.5}
\def\angle{0}
\def\Lattpts{5}
\def\lattpts{4}

\draw[fill=black!10!white]({\Radius*cos(\angle)},{\Radius*sin(\angle)})--({\Radius*cos(\angle+120)},{\Radius*sin(\angle+120)})node[midway,above,rotate=-30]{$x_1-x_3=2\pi/\alpha$}--({\Radius*cos(\angle+240)},{\Radius*sin(\angle+240)})node[midway,above,rotate=90]{$x_1-x_2=0$}--({\Radius*cos(\angle)},{\Radius*sin(\angle)})node[midway,below,rotate=30]{$x_2-x_3=0$};
\draw[dotted,fill=white]({\radius*cos(\angle)},{\radius*sin(\angle)})--({\radius*cos(\angle+120)},{\radius*sin(\angle+120)})--({\radius*cos(\angle+240)},{\radius*sin(\angle+240)})node[left]{$\vect{\rho}$}--cycle;
\draw[->,>=stealth',thick]({\radius*cos(\angle+240)},{\radius*sin(\angle+240)})
--({\radius*(cos(\angle+240)+(cos(\angle)-cos(\angle+240))/\Lattpts)},{\radius*(sin(\angle+240)+(sin(\angle)-sin(\angle+240))/\Lattpts)})node[below right]{$\vect{\omega}_1$};
\draw[->,>=stealth',thick]({\radius*cos(\angle+240)},{\radius*sin(\angle+240)})
--({\radius*(cos(\angle+240)+(cos(\angle+120)-cos(\angle+240))/\Lattpts)},{\radius*(sin(\angle+240)+(sin(\angle+120)-sin(\angle+240))/\Lattpts)})node[above left]{$\vect{\omega}_2$};

\begin{scope}[shift={{(8,0)}}]
\draw[dashed]({\Radius*cos(\angle)},{\Radius*sin(\angle)})--({\Radius*cos(\angle+120)},{\Radius*sin(\angle+120)})--({\Radius*cos(\angle+240)},{\Radius*sin(\angle+240)})--({\Radius*cos(\angle)},{\Radius*sin(\angle)});
\draw[fill=black!10!white,rotate=180,scale=.5]({\Radius*cos(\angle)},{\Radius*sin(\angle)})--({\Radius*cos(\angle+120)},{\Radius*sin(\angle+120)})node[midway,below,rotate=-30]{$x_1-x_3=\pi/\alpha$}--({\Radius*cos(\angle+240)},{\Radius*sin(\angle+240)})node[midway,below,rotate=90]{$x_1-x_2=\pi/\alpha$}--({\Radius*cos(\angle)},{\Radius*sin(\angle)})node[midway,above,rotate=30]{$x_2-x_3=\pi/\alpha$};
\draw[dotted,fill=white,rotate=180,scale=.4]({\radius*cos(\angle)},{\radius*sin(\angle)})--({\radius*cos(\angle+120)},{\radius*sin(\angle+120)})--({\radius*cos(\angle+240)},{\radius*sin(\angle+240)})node[right]{$\vect{\rho}$}--cycle;
\draw[->,>=stealth',thick,rotate=180,scale=.4]({\radius*cos(\angle+240)},{\radius*sin(\angle+240)})
--({\radius*(cos(\angle+240)+5/2*(cos(\angle)-cos(\angle+240))/\Lattpts)},{\radius*(sin(\angle+240)+5/2*(sin(\angle)-sin(\angle+240))/\Lattpts)})node[left,xshift=.3em]{$-\vect{\omega}_1$};
\draw[->,>=stealth',thick,rotate=180,scale=.4]({\radius*cos(\angle+240)},{\radius*sin(\angle+240)})
--({\radius*(cos(\angle+240)+5/2*(cos(\angle+120)-cos(\angle+240))/\Lattpts)},{\radius*(sin(\angle+240)+5/2*(sin(\angle+120)-sin(\angle+240))/\Lattpts)})node[below,yshift=.2em]{$-\vect{\omega}_2$};
\end{scope}
\end{tikzpicture}
\caption{Three-particle configuration spaces $\Sigma_{g,1}$ (white triangles with dotted sides) in the centre-of-mass plane $E_3$ with couplings yielding $M>0$ (left) and $M<0$ (right).}
\label{Fig:config}
\end{figure}

\subsection{The reduced classical Hamiltonians}\label{sec:reduced}
Next, we establish the claimed non-negativity of the products under the square roots in \eqref{Hr}, specify the value of $s(g;J)$ for $g$-values of type (i) and $J\subset\{1,\ldots,n\}$, and write down the reduced Hamiltonians $\cH_r$, $r=1,\ldots,n-1$, using root system notation.

Given any index set $J\subset\{1,\dots,n\}$, we let
\begin{equation}\label{VJ}
V_J(g;\vect{x})=\prod_{\substack{j\in J\\k\notin J}}\frac{\sin\frac{\alpha}{2}(x_j-x_k+g)}{\sin\frac{\alpha}{2}(x_j-x_k)}
\end{equation}
with $V_J\equiv 1$ if $J=\emptyset$ or $J=\{1,\dots,n\}$, and use the shorthand notation
$$
V_j\equiv V_{\{j\}},\quad j=1,\dots,n.
$$
By a direct computation, the well-known identity
$$
\prod_{\substack{j\in J\\k\notin J}}\bigg(1-\frac{\sin^2\big(\frac{\alpha g}{2}\big)}{\sin^2\frac{\alpha}{2}(x_j-x_k)}\bigg)=V_J(g;\vect{x})V_J(g;-\vect{x})
$$
is readily verified. For $p\in\{1,\dots,n-1\}$ coprime to $n$ and $g$ satisfying either of the two constraints in \eqref{typei}, we claim that $V_J(\vect{x})V_J(-\vect{x})$ is a non-negative function of $\vect{x}$ in the simplex $\overline{\Sigma}_{g,p}$ and can only vanish on the boundary. To verify this claim, we recall from \cite{FG,FKl} that
\begin{equation}\label{VnuNonneg1}
(-1)^{p-1}\sgn(M)V_j(\pm\vect{x})\geq 0,\quad \vect{x}\in\overline{\Sigma}_{g,p},
\end{equation}
with $V_j(\pm\vect{x})=0$ only along a single facet of $\overline{\Sigma}_{g,p}$, namely
\begin{equation}\label{Vjzer1}
V_j(\vect{x})=0\quad \mathrm{iff}\quad x_{j-p}-x_j=g
\end{equation}
and
\begin{equation}\label{Vjzer2}
V_j(-\vect{x})=0\quad \mathrm{iff}\quad x_j-x_{j+p}=g.
\end{equation}
We note that each $V_J$ can be expressed in terms of the $V_j$ ($j\in J$) as
\begin{equation}
V_J(\vect{x})=F_J(\vect{x})\prod_{j\in J}V_j(\vect{x}),
\label{VJ-with-Vj}
\end{equation}
where
\begin{equation}\label{FJ}
F_J(\vect{x}) = \prod_{\substack{j,k\in J\\ j\neq k}}
\frac{\sin\frac{\alpha}{2}(x_j-x_k)}{\sin\frac{\alpha}{2}(x_j-x_k+g)}
\end{equation}
cancels the extra factors from the product of the $V_j$ ($j\in J$). We clearly have
$$
F_J(-\vect{x})=F_J(\vect{x}),
$$
which together with \eqref{VJ-with-Vj} gives
$$
V_J(\vect{x})V_J(-\vect{x})=(F_J(\vect{x}))^2
\prod_{j\in J}V_j(\vect{x})V_j(-\vect{x}).
$$
Since $\vect{x}\in\overline{\Sigma}_{g,p}\subset\cA_n$ \eqref{An}, the arguments $(\alpha/2)(x_j-x_k)$ in the numerator of $F_J$ \eqref{FJ} satisfy
$$
0<\frac{\alpha}{2}|x_j-x_k|\leq\frac{\alpha}{2}(x_1-x_n)<\pi
\quad\text{for all}\quad j,k\in\{1,\dots,n\},\ j\neq k,
$$
and therefore $F_J(\vect{x})$ never vanishes for $\vect{x}\in\overline{\Sigma}_{g,p}$. Consequently, $V_J(\vect{x})$ and $V_J(-\vect{x})$ have one and the same sign within the simplex $\Sigma_{g,p}$ and only vanish on certain facets of $\overline{\Sigma}_{g,p}$, as claimed. It is now clear that the reduced Hamiltonians $\cH_r$, $r=1,\ldots,n-1$, are indeed real-valued and smooth in $\Sigma_{g,p}$.

At this point, we can define $s(g;J)$, appearing in \eqref{Hr}, as the sign of $V_J(g;\vect{x})$ (or $V_J(g;-\vect{x})$) in $\Sigma_{g,p}$, i.e.
$$
s(g;J)=\sgn(V_J(g;\vect{x})),\quad \vect{x}\in\Sigma_{g,p}.
$$
Note that \eqref{VnuNonneg1} implies that $s(\{j\})$ is independent of the value of $j\in\{1,\ldots,n\}$ and that $s(\{1,\ldots,n\})=\sgn(1)=1$.

Using the explicit formula \eqref{omegak} for the fundamental weights $\vect{\omega}_r$, it is readily verified that the reduced Hamiltonians have the following simple and uniform expression in terms of root system notation: 
\begin{equation}\label{cHr}
\cH_r(g;\vect{x},\vect{p})=\sum_{\vect{\nu}\in S_n(\vect{\omega}_r)}s(g;\vect{\nu})\cos(\langle\vect{\nu},\vect{p}\rangle)\sqrt{V_{\vect{\nu}}(g;\vect{x})V_{\vect{\nu}}(g;-\vect{x})},\quad r=1,\dots,n-1,
\end{equation}
where $S_n(\vect{\omega}_r)$ denotes the $S_n$-orbit of $\vect{\omega}_r$ under the standard action of $S_n$ on $E_n$, the coefficient functions $V_{\vect{\nu}}$ are given by 
\begin{equation}\label{Vnu}
V_{\vect{\nu}}(g;\vect{x})=
\prod_{\substack{\vect{a}\in A_{n-1}\\\langle\vect{a},\vect{\nu}\rangle=1}}
\frac{\sin\frac{\alpha}{2}(\langle\vect{a},\vect{x}\rangle+g)}{\sin\frac{\alpha}{2}\langle\vect{a},\vect{x}\rangle}
\end{equation}
and
\begin{equation}\label{snu}
s(g;\vect{\nu})=\sgn(V_{\vect{\nu}}(g;\vect{x})),\quad \vect{x}\in\Sigma_{g,p}.
\end{equation}

Finally, we note that the reduced Hamiltonians $\cH_r$, $r=1,\ldots,n-1$, are spectral invariants of the Lax matrix $L$ \cite{FG,FKl} of the compactified trigonometric RS model. Namely, one can prove that $\cH_r=((-1)^{p-1}\sgn(M))^r\mathrm{Re}(c_r)$ with $c_r$ given by $\det(L+z\mathbf{1}_n)=z^n+c_1z^{n-1}+\dots+c_{n-1}z+1$.

\subsection{The finite-dimensional Hilbert space}\label{subsec:hilbert}
We proceed to introduce the pertinent finite-dimensional Hilbert space, consisting of complex-valued functions supported on a finite uniform lattice over $\overline{\Sigma}_{g,p}$.

From the former set of inequalities in Lemma \ref{Lemma:simplex}, it is clear that $\vect{x}=\vect{\rho}_p$ is the unique `minimal' (`maximal') vertex of $\overline{\Sigma}_{g,p}$ for $M>0$ ($M<0$) in the sense that the functionals $\langle\vect{a}_{j,p},\cdot\rangle$ simultaneously assume their minimal (maximal) value at $\vect{\rho}_p$. By adding (subtracting) vectors from the cone of dominant weights $\Lambda_p^+$ to (from) $\vect{\rho}_p$, we produce the finite uniform lattice
$$
\vect{\rho}_p+\sgn(M)\Lambda^+_{p,|M|}\subset\overline{\Sigma}_{g,p}
$$
with
$$
\Lambda^+_{p,|M|}=\{\vect{\mu}\in\Lambda^+_p\mid \langle\vect{a}_{\mathrm{max},p},\vect{\mu}\rangle\leq |M|\};
$$
see Fig. \ref{Fig:lattices}. In more explicit terms, it can be described as the set of all points of the form
\begin{equation}\label{pts}
\vect{x}=\vect{\rho}_p+\sgn(M)\sum_{j=1}^{n-1}m_j\vect{\omega}_{j,p},\quad m_j\in\N_0,
\quad \sum_{j=1}^{n-1}m_j\leq |M|.
\end{equation}

By intersecting $(n-2)$ hyperplanes of the form $\langle\vect{a}_{j,p},\vect{x}-\vect{\rho}_p\rangle=0$ with the hyperplane $\langle\vect{a}_{\max,p},\vect{x}-\vect{\rho}_p\rangle=M$, we obtain the remaining vertices of $\overline{\Sigma}_{g,p}$:
$$
\vect{\rho}_p+M\vect{\omega}_{k,p},\quad k=1,\dots,n-1.
$$
These points are contained in the lattice $\vect{\rho}_p+\sgn(M)\Lambda^+_{p,|M|}$ if and only if the scale factor $\alpha$ and coupling parameter $g$ satisfy the quantisation condition
\begin{equation}\label{quant}
M\equiv \frac{2\pi}{\alpha}p-ng\in\Z\setminus\{0\}.
\end{equation}
Henceforth we shall assume that this condition is satisfied. Then we can introduce the non-negative integer $m_n\in\N_0$ by
$$
m_n=|M|-\sum_{j=1}^{n-1}m_j,
$$
and, for any lattice point $\vect{x}\in\vect{\rho}_p+\sgn(M)\Lambda_{p,|M|}^+$, we get
\begin{equation}\label{pts-2}
m_1+\dots+m_n=|M|.
\end{equation}
Note, in particular, that the boundary component $\langle\vect{a}_{\max,p},\vect{x}-\vect{\rho}_p\rangle=M$ contains all lattice points \eqref{pts} such that $m_n=0$.

We let $L^2(\vect{\rho}_p+\sgn(M)\Lambda^+_{p,|M|})$ denote the finite-dimensional Hilbert space of lattice functions $\phi\colon\vect{\rho}_p+\sgn(M)\Lambda^+_{p,|M|}\to\C$, equipped with the sesquilinear inner product
\begin{equation}\label{innerProd}
(\phi,\psi)_{p,M}=\sum_{\vect{\mu}\in\Lambda^+_{p,|M|}}\phi(\vect{\rho}_p+\sgn(M)\vect{\mu})\overline{\psi(\vect{\rho}_p+\sgn(M)\vect{\mu})}.
\end{equation}
Its dimension, which equals the cardinality of $\Lambda^+_{p,|M|}$, is given by
\begin{equation}\label{dim}
\dim\big(L^2(\vect{\rho}_p+\sgn(M)\Lambda^+_{p,|M|})\big)=\binom{n-1+|M|}{|M|}.
\end{equation}
Indeed, the points of $\Lambda^+_{p,|M|}$ can be identified with the sequences $(m_1,\dots,m_n)\in\N_0^n$ satisfying $m_1+\dots+m_n=|M|$ (c.f.~\eqref{pts-2}), and the number of such sequences is equal to the right-hand side of \eqref{dim}.

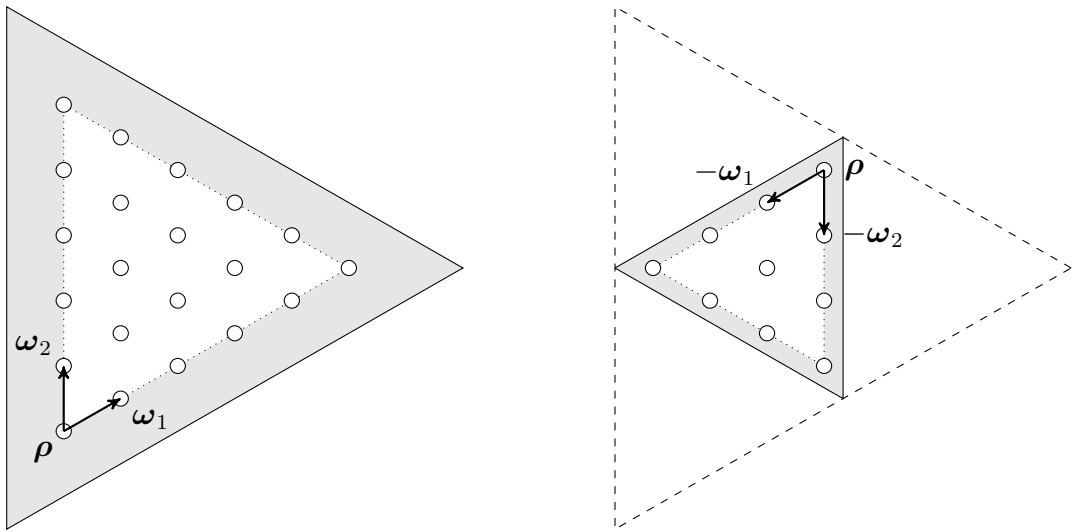
\begin{figure}[H]
\centering
\begin{tikzpicture}
\def\Radius{4}
\def\radius{2.5}
\def\cradius{.1}
\def\angle{0}
\def\Lattpts{5}
\def\lattpts{3}

\draw[fill=black!10!white]({\Radius*cos(\angle)},{\Radius*sin(\angle)})--({\Radius*cos(\angle+120)},{\Radius*sin(\angle+120)})--({\Radius*cos(\angle+240)},{\Radius*sin(\angle+240)})--cycle;
\draw[dotted,fill=white]({\radius*cos(\angle)},{\radius*sin(\angle)})--({\radius*cos(\angle+120)},{\radius*sin(\angle+120)})--({\radius*cos(\angle+240)},{\radius*sin(\angle+240)})node[below left]{$\vect{\rho}$}--cycle;
\foreach \x in {0,...,\Lattpts}{\foreach \y in {\x,...,\Lattpts}{
\draw[fill=white]
({\radius/\Lattpts*((\Lattpts-\y)*cos(\angle)+(\y-\x)*cos(\angle+120)+\x*cos(\angle+240))},
{\radius/\Lattpts*((\Lattpts-\y)*sin(\angle)+(\y-\x)*sin(\angle+120)+\x*sin(\angle+240))})circle(\cradius);}}
\draw[->,>=stealth',thick]({\radius*cos(\angle+240)},{\radius*sin(\angle+240)})
--({\radius*(cos(\angle+240)+(cos(\angle)-cos(\angle+240))/\Lattpts)},{\radius*(sin(\angle+240)+(sin(\angle)-sin(\angle+240))/\Lattpts)})node[below right]{$\vect{\omega}_1$};
\draw[->,>=stealth',thick]({\radius*cos(\angle+240)},{\radius*sin(\angle+240)})
--({\radius*(cos(\angle+240)+(cos(\angle+120)-cos(\angle+240))/\Lattpts)},{\radius*(sin(\angle+240)+(sin(\angle+120)-sin(\angle+240))/\Lattpts)})node[above left]{$\vect{\omega}_2$};

\begin{scope}[shift={{(8,0)}}]
\draw[dashed]({\Radius*cos(\angle)},{\Radius*sin(\angle)})--({\Radius*cos(\angle+120)},{\Radius*sin(\angle+120)})--({\Radius*cos(\angle+240)},{\Radius*sin(\angle+240)})--cycle;
\draw[fill=black!10!white,rotate=180,scale=.5]({\Radius*cos(\angle)},{\Radius*sin(\angle)})--({\Radius*cos(\angle+120)},{\Radius*sin(\angle+120)})--({\Radius*cos(\angle+240)},{\Radius*sin(\angle+240)})--({\Radius*cos(\angle)},{\Radius*sin(\angle)});
\draw[dotted,fill=white,rotate=180,scale=.6]({\radius*cos(\angle)},{\radius*sin(\angle)})--({\radius*cos(\angle+120)},{\radius*sin(\angle+120)})--({\radius*cos(\angle+240)},{\radius*sin(\angle+240)})node[right,xshift=.3em]{$\vect{\rho}$}--cycle;
\foreach \x in {0,...,\lattpts}{\foreach \y in {\x,...,\lattpts}{
\draw[fill=white,rotate=180,scale=.6]
({\radius/\lattpts*((\lattpts-\y)*cos(\angle)+(\y-\x)*cos(\angle+120)+\x*cos(\angle+240))},
{\radius/\lattpts*((\lattpts-\y)*sin(\angle)+(\y-\x)*sin(\angle+120)+\x*sin(\angle+240))})circle({5/3*\cradius});}}
\draw[->,>=stealth',thick,rotate=180,scale=.6]({\radius*cos(\angle+240)},{\radius*sin(\angle+240)})
--({\radius*(cos(\angle+240)+(cos(\angle)-cos(\angle+240))/\lattpts)},{\radius*(sin(\angle+240)+(sin(\angle)-sin(\angle+240))/\lattpts)})node[above left,yshift=.2em]{$-\vect{\omega}_1$};
\draw[->,>=stealth',thick,rotate=180,scale=.6]({\radius*cos(\angle+240)},{\radius*sin(\angle+240)})
--({\radius*(cos(\angle+240)+(cos(\angle+120)-cos(\angle+240))/\lattpts)},{\radius*(sin(\angle+240)+(sin(\angle+120)-sin(\angle+240))/\lattpts)})node[right,xshift=.2em]{$-\vect{\omega}_2$};
\end{scope}
\end{tikzpicture}
\caption{Three-particle lattices $\vect{\rho}+\sgn(M)\Lambda_{|M|}^+$ for $M=5$ (left) and $M=-3$ (right).}
\label{Fig:lattices}
\end{figure}

For $M>0$ an upper bound on $M$ is given by $\frac{2\pi}{\alpha q}$, corresponding to the limit $g\downarrow\frac{2\pi}{\alpha}(\frac{p}{n}-\frac{1}{nq})$. Similarly, for $M<0$ a lower bound of $-\frac{2\pi}{n-q}$ is obtained as $g\uparrow\frac{2\pi}{\alpha}(\frac{p}{n}+\frac{1}{n(n-q)})$. Therefore, to ensure that the lattice $\vect{\rho}_p+\sgn(M)\Lambda^+_{p,|M|}$ does not consist of only a single point $\vect{\rho}_p$, i.e.~that, for all type (i) couplings $g$ \eqref{typei}, we have $|M|\geq 1$, the scale parameter $\alpha>0$ has to satisfy
\begin{equation}\label{bound-on-alpha}
\frac{2\pi}{\alpha}>q,\quad\text{if}\quad M>0\qquad\text{and}\qquad
\frac{2\pi}{\alpha}>n-q,\quad \text{if}\quad M<0.
\end{equation}

\subsection{The discrete difference operators}\label{subsec:operators}
In a pioneering paper, Ruijsenaars \cite{Rui87} established a formal canonical quantisation procedure for classical relativistic Calogero-Moser systems that preserves integrability. (Here the word `formal' refers to the fact that no Hilbert space is specified.) In the present case, the procedure yields the pair-wise commuting difference operators
\begin{equation}\label{hcHr}
\hat{\cS}_r(g)\equiv \sum_{\vect{\nu}\in S_n(\vect{\omega}_r)} V_{\vect{\nu}}^{1/2}(g;\vect{x}) \exp(\langle\vect{\nu},\partial/\partial \vect{x}\rangle) V_{\vect{\nu}}^{1/2}(g;-\vect{x}),\quad r=1,\dots,n-1,
\end{equation}
with $V_{\vect{\nu}}$ given by \eqref{Vnu} and $\exp(\langle\vect{\nu},\partial/\partial \vect{x}\rangle)$, $\vect{\nu}\in\Lambda$, acting on functions $\phi\colon E_n\to\C$ according to
$$
\big(\exp(\langle\vect{\nu},\partial/\partial \vect{x}\rangle)\phi\big)(\vect{x})=\phi(\vect{x}+\vect{\nu}),
$$
and where $g$ is assumed to belong to a type (i) interval \eqref{typei} and we choose the square roots such that
\begin{equation}\label{sr}
\sgn\big(V_{\vect{\nu}}^{1/2}(g;\vect{x})V_{\vect{\nu}}^{1/2}(g;-\vect{x}-\vect{\nu})\big)=s(g;\vect{\nu}),\quad \vect{x},\vect{x}+\vect{\nu}\in\Sigma_{g,p},
\end{equation}
c.f.~\eqref{snu}. (To be precise, Ruijsenaars focused on the more general elliptic case and considered somewhat different operators $\hat{S}_r$, with $\hat{\cS}_r$ corresponding to $\hat{S}_{r}\hat{S}_n^{-r/n}$.)

Using $\vect{\omega}_{n-r}\in S_n(-\vect{\omega}_r)$ and the identity $V_{-\vect{\nu}}(\vect{x})=V_{\vect{\nu}}(-\vect{x})$, we deduce
\begin{equation}\label{hcHn-r}
\hat{\cS}_{n-r}=\sum_{\vect{\nu}\in S_n(\vect{\omega}_r)} V_{\vect{\nu}}^{1/2}(-\vect{x}) \exp(-\langle\vect{\nu},\partial/\partial \vect{x}\rangle) V_{\vect{\nu}}^{1/2}(\vect{x}),
\end{equation}
which, in particular, implies that $\hat{\cS}_{n-r}$ is the formal adjoint of $\hat{\cS}_r$. To each classical Hamiltonian $\cH_r$ it is thus natural to associate the quantum Hamiltonian $\hat{\cH}_r\equiv(\hat{\cS}_r+\hat{\cS}_{n-r})/2$. Indeed, $\hat{\cH}_r$ is clearly formally self-adjoint and \eqref{sr} ensures that $\cH_r$ is recovered by substituting $\partial/\partial \vect{x}\to\mathrm{i}\vect{p}$. For $p=1$ and $M>0$, these are precisely the quantum Hamiltonians considered by van Diejen and Vinet.

As we demonstrate below, the difference operators $\hat{\cH}_r$ are the appropriate quantum Hamiltonians whenever $g$ belongs to a type (i) parameter-interval \eqref{typei} with $M>0$, but the $g$-intervals with $M<0$ require a slight modification. More precisely, replacing $\hat{\cS}_r$ by
\begin{equation}\label{hcHrM}
\hat{\cS}_{r,M}(g)\equiv \sum_{\vect{\nu}\in S_n(\vect{\omega}_r)} V_{\vect{\nu}}^{1/2}(g;\vect{x}) \exp(\sgn(M)\langle\vect{\nu},\partial/\partial \vect{x}\rangle) V_{\vect{\nu}}^{1/2}(g;-\vect{x}),
\end{equation}
where $r=1,\dots,n-1$, we shall consider the difference operators
\begin{equation}\label{hHrM}
\hat{\cH}_{r,M}\equiv \frac{1}{2}(\hat{\cS}_{r,M}+\hat{\cS}_{n-r,M}),\quad r=1,\dots,n-1.
\end{equation}
We observe the identity $V_{\vect{\nu}}(-g;\vect{x})=V_{\vect{\nu}}(g;-\vect{x})$ and extend the definition \eqref{snu} of $s(g;\vect{\nu})$ (and hence of $\cH_r$ and $\hat{\cS}_r, \hat{\cH}_r$) to $g<0$ by
$$
s(-g;\vect{\nu})\equiv \sgn(V_{\vect{\nu}}(-g;\vect{x}))=\sgn(V_{\vect{\nu}}(g;-\vect{x}))=s(g;\vect{\nu}),
$$
whenever $g$ belongs to a type (i) interval \eqref{typei}. Then \eqref{hcHn-r} and the identity
\begin{equation}\label{VnuRefl}
V_{\vect{\nu}}(-g;-\vect{x})=V_{\vect{\nu}}(g;\vect{x})
\end{equation}
imply
\begin{equation}\label{hcHrMtohcHr}
\hat{\cS}_{r,M}(g)
=\begin{cases}
\hat{\cS}_r(g),&M>0\\
\hat{\cS}_{n-r}(-g),&M<0
\end{cases}
\end{equation}
so that $\hat{\cH}_{r,M}(g)=\hat{\cH}_r(\sgn(M)g)$. We thus retain commutativity, formal self-adjointness and the classical analogue of $\hat{\cH}_{r,M}$ again amounts to $\cH_r$, since $\cH_r(-g)=\cH_r(g)$.

We now turn to the problem of promoting the formal difference operators \eqref{hHrM} to well-defined self-adjoint operators in the Hilbert space $L^2(\vect{\rho}_p+\sgn(M)\Lambda^+_{p,|M|})$. The following technical lemma will be the key ingredient. We remark that condition 3 of the lemma is readily bypassed in the subsequent analysis and our main result does not depend on it.

\begin{lemma}\label{Lemma:Vnu}
Let $p\in\{1,\dots,n-1\}$ be coprime to $n$, and choose $\alpha,g>0$ such that
\begin{enumerate}
\item the coupling parameter $g$ belongs to one of the two type (i) parameter-intervals \eqref{typei},
\item the quantisation condition \eqref{quant} holds true.
\end{enumerate}
Then, for $\vect{\nu}\in S_n(\vect{\omega}_r)$ with $r=1,\dots,n-1$ and $\vect{\mu}\in\Lambda^+_{p,|M|}$, we have
\begin{equation}\label{VnuIneqs}
0<V_{\vect{\nu}}(\vect{\rho}_p+\sgn(M)\vect{\mu})V_{\vect{\nu}}(-\vect{\rho}_p-\sgn(M)(\vect{\mu}+\vect{\nu}))<\infty\quad\text{if}\quad\vect{\mu}+\vect{\nu}\in\Lambda^+_{p,|M|}
\end{equation}
and
\begin{equation}\label{Vzer}
V_{\vect{\nu}}(\vect{\rho}_p+\sgn(M)\vect{\mu})=0\quad\text{if}\quad \vect{\mu}+\vect{\nu}\notin\Lambda^+_{p,|M|}.
\end{equation}
If, in addition, we prescribe that
\begin{enumerate}
\item[3.] the coupling parameter $g\notin\bigg\{\dfrac{1-\sgn(M)A+B\frac{2\pi}{\alpha}}{C}\,\bigg\vert\,
\begin{smallmatrix}
A&=&0,\dots,|M|\\
B&=&0,\dots,p-1\\
C&=&1,\dots,n-1
\end{smallmatrix}\bigg\}$,
\end{enumerate}
then we also ensure that
\begin{equation}\label{Vbound}
-\infty<V_{\vect{\nu}}(-\vect{\rho}_p-\sgn(M)(\vect{\mu}+\vect{\nu}))<\infty\quad \text{if}\quad\vect{\mu}+\vect{\nu}\notin\Lambda^+_{p,|M|}.
\end{equation}
\end{lemma}

\begin{proof}
To $\vect{\nu}\in S_n(\vect{\omega}_r)$ we associate a unique index set
$$
J\subset\{1,\dots,n\},\quad |J|=r,
$$
by requiring
$$
\vect{\nu}=\sum_{j\in J}\vect{e}_j-\frac{r}{n}(\vect{e}_1+\cdots+\vect{e}_n).
$$
Recalling the function $F_J$ \eqref{FJ}, we observe that
\begin{equation}\label{VnuRep}
V_{\vect{\nu}}(\vect{x})=F_J(\vect{x})\prod_{j\in J}V_j(\vect{x}).
\end{equation}
Note that $F_J(\vect{x})$ cancels the unwanted factors $\sin(\alpha(x_j-x_k+g)/2)/\sin(\alpha(x_j-x_k)/2)$ with $k\in J$ in $V_j(\vect{x})$. Following the proof of Lemma \ref{Lemma:simplex}, we rewrite the vanishing conditions \eqref{Vjzer1} and \eqref{Vjzer2} in terms of the simple roots \eqref{ajp} and the maximal root \eqref{amax}:
\begin{equation}\label{Vjzer3}
V_j(\vect{x})=0\quad \mathrm{iff}\quad
\begin{cases}\langle\vect{a}_{j-p,p},\vect{x}-\vect{\rho}_p\rangle=0,&j\neq p,\\ \langle\vect{a}_{\mathrm{max},p},\vect{x}-\vect{\rho}_p\rangle=M,&j=p,
\end{cases}
\end{equation}
and
\begin{equation}\label{Vjzer4}
V_j(-\vect{x})=0\quad \mathrm{iff}\quad
\begin{cases}\langle\vect{a}_{j,p},\vect{x}-\vect{\rho}_p\rangle=0,&j\neq n,\\ \langle\vect{a}_{\mathrm{max},p},\vect{x}-\vect{\rho}_p\rangle=M,& j=n,
\end{cases}
\end{equation}
where $\vect{a}_{j-p,p}\equiv \vect{a}_{n-p+j,p}$ for $j=1,\dots,p-1$.

From Eqs. \eqref{VnuRep}--\eqref{Vjzer3}, we infer that $V_{\vect{\nu}}(\vect{\rho}_p+\sgn(M)\vect{\mu})$ vanishes precisely when either $\langle\vect{a}_{\mathrm{max},p},\vect{\nu}\rangle=1$ and $\langle\vect{a}_{\mathrm{max},p},\vect{\mu}\rangle=|M|$ or $\langle\vect{a}_{j,p},\vect{\nu}\rangle=-1$ and $\langle\vect{a}_{j,p},\vect{\mu}\rangle=0$ for some $j=1,\dots,n-1$. By Lemma \ref{Lemma:simplex}, this is equivalent to $\vect{\mu}\in\Lambda^+_{p,|M|}$ and $\vect{\mu}+\vect{\nu}\notin\Lambda^+_{p,|M|}$. Similarly, from Eqs. \eqref{VnuRep},\eqref{Vjzer4} we find that $V_{\vect{\nu}}(-\vect{\rho}_p-\sgn(M)(\vect{\mu}+\vect{\nu}))$ vanishes only when either $\langle\vect{a}_{\mathrm{max},p},\vect{\nu}\rangle=-1$ and $\langle\vect{a}_{\mathrm{max},p},\vect{\mu}+\vect{\nu}\rangle=|M|$ or $\langle\vect{a}_{j,p},\vect{\nu}\rangle=1$ and $\langle\vect{a}_{j,p},\vect{\mu}+\vect{\nu}\rangle=0$ with $j=1,\dots,n-1$, neither of which occur for $\vect{\mu},\vect{\mu}+\vect{\nu}\in\Lambda^+_{p,|M|}$.

Hence the statement will follow once we can show that the product in \eqref{VnuIneqs} is finite and non-negative. Due to $\overline{\Sigma}_{g,p}\subset\cA_n$, we have $0<|\langle\vect{a},\vect{x}\rangle|<2\pi/\alpha$ for all $\vect{a}\in A_{n-1}$ and $\vect{x}\in\overline{\Sigma}_{g,p}$, which entails
\begin{equation}\label{VnuFin}
|V_{\vect{\nu}}(\pm\vect{x})|<\infty,\quad \vect{x}\in\overline{\Sigma}_{g,p}.
\end{equation}
By combining \eqref{VnuRep} with the invariance properties
$$
F_J(-\vect{x})=F_J(\vect{x}),\qquad F_J(\vect{x}+\vect{\nu})=F_J(\vect{x}),
$$
we deduce that
\begin{equation}\label{VnuProd}
V_{\vect{\nu}}(\vect{x})V_{\vect{\nu}}(-\vect{x}\pm\vect{\nu})=(F_J(\vect{x}))^2\prod_{j\in J}V_j(\vect{x})V_j(-\vect{x}\pm\vect{\nu}).
\end{equation}
By assumption $\vect{\mu},\vect{\mu}+\vect{\nu}\in\Lambda^+_{p,|M|}$, so finiteness is clear from \eqref{VnuFin}, whereas non-negativity is a simple consequence of \eqref{VnuProd} and \eqref{VnuNonneg1}.

To ensure that the coefficients $V_{\vect{\nu}}(-\vect{x}-\sgn(M)\vect{\nu})$ are finite when $\vect{x}\in\vect{\rho}_p+\sgn(M)\Lambda^+_{p,|M|}$ is such a boundary lattice point that $\vect{x}+\sgn(M)\vect{\nu}\notin\vect{\rho}_p+\sgn(M)\Lambda^+_{p,|M|}$,
certain values of $g$ must be excluded. We find these values by splitting up the factors in $V_{\vect{\nu}}(-\vect{x}-\sgn(M)\vect{\nu})$ into two groups corresponding to positive and negative roots in $A_{n-1}$. This yields
\begin{multline*}
V_{\vect{\nu}}(-\vect{x}-\sgn(M)\vect{\nu})\\
=\prod_{\substack{\vect{a}\in A_{n-1}^+\\\langle\vect{a},\vect{\nu}\rangle=1}}
\frac{\sin\frac{\alpha}{2}(\langle\vect{a},-\vect{x}-\sgn(M)\vect{\nu}\rangle+g)}{\sin\frac{\alpha}{2}\langle\vect{a},-\vect{x}-\sgn(M)\vect{\nu}\rangle}
\prod_{\substack{\vect{a}\in A_{n-1}^+\\\langle\vect{a},\vect{\nu}\rangle=-1}}
\frac{\sin\frac{\alpha}{2}(\langle\vect{a},-\vect{x}-\sgn(M)\vect{\nu}\rangle-g)}{\sin\frac{\alpha}{2}\langle\vect{a},-\vect{x}-\sgn(M)\vect{\nu}\rangle}\\
=\prod_{\substack{\vect{a}\in A_{n-1}^+\\\langle\vect{a},\vect{\nu}\rangle=\sgn(M)}}
\frac{\sin\frac{\alpha}{2}(\langle\vect{a},\vect{x}\rangle+1-\sgn(M)g)}{\sin\frac{\alpha}{2}(\langle\vect{a},\vect{x}\rangle+1)}
\prod_{\substack{\vect{a}\in A_{n-1}^+\\\langle\vect{a},\vect{\nu}\rangle=-\sgn(M)}}
\frac{\sin\frac{\alpha}{2}(\langle\vect{a},\vect{x}\rangle-1+\sgn(M)g)}{\sin\frac{\alpha}{2}(\langle\vect{a},\vect{x}\rangle-1)}.
\end{multline*}
Since $2\pi/\alpha>1$ (cf. \eqref{bound-on-alpha}) and $0<\langle\vect{a},\vect{x}\rangle<2\pi/\alpha$ for any $\vect{a}\in A_{n-1}^+$ and $\vect{x}\in\overline{\Sigma}_{g,p}$, the denominators can vanish only if
\begin{equation}\label{denom-vanish}
\langle\vect{a},\vect{x}\rangle-1=0
\quad\text{or}\quad
\langle\vect{a},\vect{x}\rangle+1=\frac{2\pi}{\alpha}.
\end{equation}
Let us fix a positive root $\vect{a}=\vect{e}_j-\vect{e}_k\in A_{n-1}^+$, where $1\leq j<k\leq n$. Given that $p$ is coprime to $n$, there exists a number $\ell_{j,k}\in\{1,\dots,n-1\}$ such that
$$
j+\ell_{j,k}p\equiv k\pmod{n}.
$$
Consequently, using the periodicity convention $\vect{e}_{n+j}\equiv \vect{e}_j$, we have
\begin{equation}\label{posroot-with-p-simple}
\vect{a}=(\vect{e}_j-\vect{e}_{j+p})+(\vect{e}_{j+p}-\vect{e}_{j+2p})+\dots+(\vect{e}_{j+(\ell_{j,k}-1)p}-\vect{e}_{j+\ell_{j,k}p})=\sum_{i\in I}\vect{a}_{i,p}
\end{equation}
with index set
$$
I\equiv \{j+(\ell-1)p\mod{n}\mid\ell=1,\dots,\ell_{j,k}\}\subsetneq\{1,\dots,n\}.
$$
Introducing the additional index set
$$
K\equiv \{i\in I\mid n-p<i\leq n\},
$$
and keeping in mind that each lattice point $\vect{x}\in\vect{\rho}_p+\sgn(M)\Lambda_{p,|M|}^+$ is of the form \eqref{pts}, lets us express the first equation in \eqref{denom-vanish} as
$$
|I|g-|K|\frac{2\pi}{\alpha}+\sgn(M)\sum_{i\in I}m_i-1=0.
$$
By solving this equation for $g$, we get
\begin{equation}\label{solve-for-g-1}
g=\frac{1-\sgn(M)\sum_{i\in I}m_i+|K|\frac{2\pi}{\alpha}}{|I|}.
\end{equation}
Letting $I'=\{1,\dots,n\}\setminus I$, we have
$$
\sum_{i\in I}m_i+\sum_{i'\in I'}m_{i'}=m_1+\dots+m_n=|M|,
$$
and therefore we can recast the second equation in \eqref{denom-vanish} as
$$
|I|g-|K|\frac{2\pi}{\alpha}+M-\sgn(M)\sum_{i'\in I'}m_{i'}+1=\frac{2\pi}{\alpha}.
$$
Plugging $M=2\pi p/\alpha-ng$ into this equation and solving for $g$ yields
\begin{equation}\label{solve-for-g-2}
g=\frac{1-\sgn(M)\sum_{i'\in I'}m_{i'}+(p-1-|K|)\frac{2\pi}{\alpha}}{n-|I|}.
\end{equation}
We note that the cardinality $|K|$ of $K$ counts the number of times $j+\ell p$ `steps over' a multiple of $n$ as $\ell$ is increased from $1$ to $\ell_{j,k}$, which entails
$$
|K|=\frac{j-k+\ell_{j,k}p}{n}\leq \frac{j-(j+1)+(n-1)p}{n}<p.
$$
Thus, we have $1\leq |I|<n$, $0\leq|K|<p$, $0\leq\sum_{i\in I}m_i\leq|M|$ in \eqref{solve-for-g-1} and \eqref{solve-for-g-2}, which implies that \eqref{denom-vanish} can be satisfied only if
\begin{equation}\label{excluded}
g\in\bigg\{\frac{1-\sgn(M)A+B\frac{2\pi}{\alpha}}{C}\,\bigg\vert\,
\begin{smallmatrix}
A&=&0,\dots,|M|\\
B&=&0,\dots,p-1\\
C&=&1,\dots,n-1
\end{smallmatrix}\bigg\}.
\end{equation}
Therefore, imposing Condition 3 ensures that the $V_{\vect{\nu}}(-\vect{\rho}_p-\sgn(M)(\vect{\mu}+\vect{\nu}))$ are finite if $\vect{\mu}+\vect{\nu}\notin\Lambda^+_{p,|M|}$. This concludes the proof.
\end{proof}

It is worth mentioning that in the standard case $0<g<2\pi/\alpha n$, when $p=1$ and $M>0$, the excluded values of $g$ are obtained by setting $A=0$ and $B=0$ in \eqref{excluded}: $g\in\{1/C\mid C=1,\dots,n-1\}$.

Assuming that $p\perp n$ and $\alpha,g>0$ satisfy Conditions 1--3 of Lemma \ref{Lemma:Vnu} ensures that the coefficient functions
\begin{equation}\label{W}
W_{\vect{\nu}}(\vect{x})\equiv V^{1/2}_{\vect{\nu}}(\vect{x})V^{1/2}_{\vect{\nu}}(-\vect{x}-\sgn(M)\vect{\nu}),
\end{equation}
cf.~\eqref{hcHrM}, take non-zero real values at all lattice points $\vect{x}=\vect{\rho}_p+\sgn(M)\vect{\mu}$ with $\vect{\mu}\in\Lambda^+_{p,|M|}$ and $\vect{\mu}+\vect{\nu}\in\Lambda^+_{p,|M|}$, and vanish whenever $\vect{\mu}+\vect{\nu}\notin\Lambda^+_{p,|M|}$. By \eqref{sr} and \eqref{snu}, we have
$$
\sgn\big(W_{\vect{\nu}}(\vect{\rho}_p+\sgn(M)\vect{\mu})\big)=\sgn\big(V_{\vect{\nu}}(\vect{\rho}_p+\sgn(M)\vect{\mu})\big).
$$
In the special case $p=1$, $M>0$, van Diejen and Vinet \cite[Lemma 3.2]{vDV} proved that $V_{\vect{\nu}}(\vect{x})$ is always positive in the simplex $\Sigma_{g,p}$ independent of $\vect{\nu}$, but this is not true in general. Instead, each $V_{\vect{\nu}}(\vect{x})$ has a definite sign in the simplex $\Sigma_{g,p}$ that depends on the choice of $\vect{\nu}$. For example, consider the case $n=4$, $p=1$ and $M<0$ with $\vect{\nu}=\vect{\omega}_2$ and $\vect{\nu}'=\vect{\omega}_1-\vect{\omega}_2+\vect{\omega}_3\in S_4(\vect{\omega}_2)$. A direct computation, involving the defining inequalities of the simplex, shows that
$$
V_{\vect{\nu}}(\vect{x})<0<V_{\vect{\nu}'}(\vect{x}),\quad \forall\,\vect{x}\in\Sigma_{g,1}.
$$

Fixing a lattice function $\phi\in L^2(\vect{\rho}_p+\sgn(M)\Lambda^+_{p,|M|})$, let us extend the domain of $\phi$ by specifying the values
$$
\phi(\vect{\rho}_p+\sgn(M)(\vect{\mu}+\vect{\nu}))\in\C,\qquad \vect{\mu}\in \Lambda^+_{p,|M|},\ \ \vect{\mu}+\vect{\nu}\notin \Lambda^+_{p,|M|},\ \ \vect{\nu}\in S_n(\vect{\omega}_r),
$$
where $r=1,\dots,n-1$. Independent of the specific values we assign, \eqref{Vzer}--\eqref{Vbound} entail that
\begin{equation}\label{hcHrMAct}
\big(\hat\cS_{r,M} \phi\big)(\vect{\rho}_p+\sgn(M)\vect{\mu})=\sum_{\substack{\vect{\nu}\in S_n(\vect{\omega}_r)\\ \vect{\mu}+\vect{\nu}\in\Lambda^+_{p,|M|}}}W_{\vect{\nu}}(\vect{\rho}_p+\sgn(M)\vect{\mu})\phi(\vect{\rho}_p+\sgn(M)(\vect{\mu}+\vect{\nu})).
\end{equation}
We stress that the right-hand side depends only on the values of $\phi$ on the lattice $\vect{\rho}_p+\sgn(M)\Lambda^+_{p,|M|}$. In this sense, each difference operator $\hat\cS_{r,M}$ admits a well-defined restriction onto the Hilbert space $L^2(\vect{\rho}_p+\sgn(M)\Lambda^+_{p,|M|})$, whose action on a lattice function $\phi$ is given explicitly by \eqref{hcHrMAct}. If $g$ belongs to the set of excluded values \eqref{excluded}, $V_{\vect{\nu}}(-\vect{\rho}_p-\sgn(M)(\vect{\mu}+\vect{\nu}))$ may be singular at a boundary point $\vect{\mu}\in\Lambda^+_{p,|M|}$ with $\vect{\mu}+\vect{\nu}\notin\Lambda^+_{p,|M|}$, thus rendering the coefficient function $W_{\vect{\nu}}(\vect{\rho}_p+\sgn(M)\vect{\mu})$ ill-defined. We shall resolve this issue by insisting on continuity in $g$, so that \eqref{hcHrMAct} holds true regardless of whether or not Condition 3 is satisfied.

As the following proposition demonstrates, the corresponding restriction of each difference operator $\hat{\cH}_{r,M}$ \eqref{hHrM} yields a self-adjoint operator in $L^2(\vect{\rho}_p+\sgn(M)\Lambda^+_{p,|M|})$.

\begin{proposition}\label{Prop:adj}
Assume that $p\perp n$ and Conditions 1--2 in Lemma \ref{Lemma:Vnu} on the parameters $\alpha,g>0$ hold true. Let \eqref{hcHrMAct} define the action of the operators $\hat{\cS}_{r,M}\colon L^2(\vect{\rho}_p+\sgn(M)\Lambda^+_{p,|M|})\to L^2(\vect{\rho}_p+\sgn(M)\Lambda^+_{p,|M|})$. Then we have
\begin{equation}\label{duality}
\big(\hat{\cS}_{r,M}\phi,\psi\big)_{p,M}=\big(\phi,\hat{\cS}_{n-r,M}\psi\big)_{p,M}
\end{equation}
for all $r=1,\dots,n-1$ and $\phi,\psi\in L^2(\vect{\rho}_p+\sgn(M)\Lambda^+_{p,|M|})$. Consequently, the operators $\hat{\cH}_{r,M}\equiv\frac{1}{2}(\hat{\cS}_{r,M}+\hat{\cS}_{n-r,M})$, $r=1,\dots,n-1$, are self-adjoint in $L^2(\vect{\rho}_p+\sgn(M)\Lambda^+_{p,|M|})$.
\end{proposition}

\begin{proof}
In order to shorten the formulae below, we first consider the case $M>0$. By \eqref{innerProd} and \eqref{hcHrMAct}, we then have
\begin{multline*}
\big(\hat{\cS}_{r,M}\phi,\psi\big)_{p,M} = \sum_{\vect{\mu}\in\Lambda^+_{p,|M|}}\big(\hat{\cS}_{r,M}\phi\big)(\vect{\rho}_p+\vect{\mu})\overline{\psi(\vect{\rho}_p+\vect{\mu})}\\
= \sum_{\vect{\nu}\in S_n(\vect{\omega}_r)}\sum_{\substack{\vect{\mu}\in\Lambda^+_{p,|M|}\\ \vect{\mu}+\vect{\nu}\in\Lambda^+_{p,|M|}}}W_{\vect{\nu}}(\vect{\rho}_p+\vect{\mu})\phi(\vect{\rho}_p+\vect{\mu}+\vect{\nu})\overline{\psi(\vect{\rho}_p+\vect{\mu})}\\
= \sum_{\vect{\nu}\in S_n(\vect{\omega}_r)}\sum_{\substack{\vect{\mu}^\prime\in\Lambda^+_{p,|M|}\\ \vect{\mu}^\prime-\vect{\nu}\in\Lambda^+_{p,|M|}}}W_{\vect{\nu}}(\vect{\rho}_p+\vect{\mu}^\prime-\vect{\nu})\phi(\vect{\rho}_p+\vect{\mu}^\prime)\overline{\psi(\vect{\rho}_p+\vect{\mu}^\prime-\vect{\nu})}.
\end{multline*}
Invoking the identity
$$
W_{\vect{\nu}}(\vect{\rho}_p+\vect{\mu}^\prime-\vect{\nu})=W_{-\vect{\nu}}(\vect{\rho}_p+\vect{\mu}^\prime)
$$
and using $\vect{\omega}_{n-r}\in S_n(-\vect{\omega}_r)$, we obtain
\begin{multline*}
\sum_{\vect{\nu}\in S_n(\vect{\omega}_{n-r})}\sum_{\substack{\vect{\mu}^\prime\in\Lambda^+_{p,|M|}\\ \vect{\mu}^\prime+\vect{\nu}\in\Lambda^+_{p,|M|}}}\phi(\vect{\rho}_p+\vect{\mu}^\prime)\overline{W_{\vect{\nu}}(\vect{\rho}_p+\vect{\mu}^\prime)\psi(\vect{\rho}_p+\vect{\mu}^\prime+\vect{\nu})}\\
=\sum_{\vect{\mu}^\prime\in\Lambda^+_{p,|M|}}\phi(\vect{\rho}_p+\vect{\mu}^\prime)\overline{\big(\hat{\cS}_{n-r,M}\psi\big)(\vect{\rho}_p+\vect{\mu}^\prime)}=\big(\phi,\hat{\cS}_{n-r,M}\psi\big)_{p,M}.
\end{multline*}
The same reasoning applies with minor changes to the case $M<0$. (In the above formulae, we need only change the sign of $\vect{\mu},\vect{\mu}^\prime$ and $\vect{\nu}$ in the arguments of the functions $W_{\pm\vect{\nu}}$, $\phi$ and $\psi$.)
\end{proof}

\section{Joint eigenfunctions}\label{sec:eigfuncs}
The main purpose of this section is to construct an orthonormal basis in the Hilbert space $L^2(\vect{\rho}_p+\sgn(M)\Lambda^+_{p,|M|})$ consisting of joint eigenfunctions of the discrete difference operators $\hat{\cS}_{r,M}$, $r=1,\dots,n-1$. In the process, the corresponding joint eigenvalues will be obtained explicitly.

\subsection{A factorised joint eigenfunction}\label{subsec:ground-top-state}
In a first step, we identify a positive joint eigenfunction $\Psi_{\vect{0},p}$ of factorised form. From Eq.~(4.3) in \cite{vDV}, we recall the trigonometric Pochhammer symbol
$$
(z:\sin_\alpha)_m\equiv
\begin{cases}
1,& \text{if}\ m=0\\[2pt]
\sin\frac{\alpha}{2}(z)\cdots \sin\frac{\alpha}{2}(z+m-1), & \text{if}\ m=1,2,\dots \\[2pt]
\dfrac{1}{\sin\frac{\alpha}{2}(z-1)\cdots \sin\frac{\alpha}{2}(z+m)}, & \text{if}\ m=-1,-2,\dots
\end{cases}
$$
which satisfies the easily verified identity
\begin{equation}\label{Pochhammer-id}
\frac{(z:\sin_\alpha)_{m+1}}{(z:\sin_\alpha)_m}=\sin\frac{\alpha}{2}(z+m),\quad m\in\Z.
\end{equation}
The key ingredient is the lattice function $\Delta_p\colon\Lambda_{p,|M|}^+\to\R$ given by
\begin{equation}\label{Delta}
\Delta_p(\vect{\mu})=\prod_{\vect{a}\in A_{n-1,p}^+}\frac{\sin\frac{\alpha}{2}\langle\vect{a},\vect{\rho}_p+\sgn(M)\vect{\mu}\rangle}{\sin\frac{\alpha}{2}\langle\vect{a},\vect{\rho}_p\rangle}\frac{(\langle\vect{a},\vect{\rho}_p\rangle+\sgn(M)g:\sin_\alpha)_{\sgn(M)\langle\vect{a},\vect{\mu}\rangle}}{(\langle\vect{a},\vect{\rho}_p\rangle+1-\sgn(M)g:\sin_\alpha)_{\sgn(M)\langle\vect{a},\vect{\mu}\rangle}}.
\end{equation}
Its basic properties (regularity and positivity) are readily established using the recurrence relations detailed in the following lemma.

\begin{lemma}\label{lemma:rec-rel}
For any $\vect{\mu}\in\Lambda_{p,|M|}^+$ and $\vect{\nu}\in S_n(\vect{\omega}_r)$, $r=1,\dots,n-1$ satisfying $\vect{\mu}+\vect{\nu}\in\Lambda_{p,|M|}^+$, we have
\begin{equation}\label{rec-rel}
\frac{\Delta_p(\vect{\mu}+\vect{\nu})}{\Delta_p(\vect{\mu})}
=\frac{V_{\vect{\nu}}(\vect{\rho}_p+\sgn(M)\vect{\mu})}{V_{\vect{\nu}}(-\vect{\rho}_p-\sgn(M)(\vect{\mu}+\vect{\nu}))}.
\end{equation}
\end{lemma}

\begin{proof}
Let us consider the left-hand side of \eqref{rec-rel}. In the numerator we have $\Delta_p(\vect{\mu}+\vect{\nu})$, which differs from $\Delta_p(\vect{\mu})$ \eqref{Delta} only in factors that have $\langle\vect{a},\vect{\nu}\rangle=\pm 1$. Hence, by applying \eqref{Pochhammer-id}, the ratio $\Delta_p(\vect{\mu}+\vect{\nu})/\Delta_p(\vect{\mu})$ can be rewritten as
\begin{multline*}
\frac{\Delta_p(\vect{\mu}+\vect{\nu})}{\Delta_p(\vect{\mu})}\\
=\prod_{\substack{\vect{a}\in A_{n-1,p}^+\\\langle\vect{a},\vect{\nu}\rangle=1}}
\frac{\sin\frac{\alpha}{2}\langle\vect{a},\vect{\rho}_p+\sgn(M)(\vect{\mu}+\vect{\nu})\rangle}{\sin\frac{\alpha}{2}\langle\vect{a},\vect{\rho}_p+\sgn(M)\vect{\mu}\rangle}
\frac{\sin\frac{\alpha}{2}(\langle\vect{a},\vect{\rho}_p+\sgn(M)\vect{\mu}\rangle+g)}{\sin\frac{\alpha}{2}(\langle\vect{a},\vect{\rho}_p+\sgn(M)(\vect{\mu}+\vect{\nu})\rangle-g)}\\
\times\prod_{\substack{\vect{a}\in A_{n-1,p}^+\\\langle\vect{a},\vect{\nu}\rangle=-1}}
\frac{\sin\frac{\alpha}{2}\langle\vect{a},\vect{\rho}_p+\sgn(M)(\vect{\mu}+\vect{\nu})\rangle}{\sin\frac{\alpha}{2}\langle\vect{a},\vect{\rho}_p+\sgn(M)\vect{\mu}\rangle}
\frac{\sin\frac{\alpha}{2}(\langle\vect{a},\vect{\rho}_p+\sgn(M)\vect{\mu}\rangle-g)}{\sin\frac{\alpha}{2}(\langle\vect{a},\vect{\rho}_p+\sgn(M)(\vect{\mu}+\vect{\nu})\rangle+g)}.
\end{multline*}
Recombining the two products yields
\begin{multline*}
\frac{\Delta_p(\vect{\mu}+\vect{\nu})}{\Delta_p(\vect{\mu})}\\
=\prod_{\substack{\vect{a}\in A_{n-1}\\\langle\vect{a},\vect{\nu}\rangle=1}}
\frac{\sin\frac{\alpha}{2}(\langle\vect{a},\vect{\rho}_p+\sgn(M)\vect{\mu}\rangle+g)}
{\sin\frac{\alpha}{2}\langle\vect{a},\vect{\rho}_p+\sgn(M)\vect{\mu}\rangle}
\frac{\sin\frac{\alpha}{2}\langle\vect{a},-\vect{\rho}_p-\sgn(M)(\vect{\mu}+\vect{\nu})\rangle}{\sin\frac{\alpha}{2}(\langle\vect{a},-\vect{\rho}_p-\sgn(M)(\vect{\mu}+\vect{\nu})\rangle+g)}\\
=\frac{V_{\vect{\nu}}(\vect{\rho}_p+\sgn(M)\vect{\mu})}{V_{\vect{\nu}}(-\vect{\rho}_p-\sgn(M)(\vect{\mu}+\vect{\nu}))}.\qquad\qquad\qquad\qquad\qquad\qquad\qquad\qquad\qquad\quad
\end{multline*}
This concludes the proof.
\end{proof}

From \eqref{Delta}, it is clear that $\Delta_p(\vect{0})=1$. Starting from $\vect{\mu}=\vect{0}$, any lattice point $\vect{\mu}=\sum_{j=1}^{n-1}m_j\vect{\omega}_{j,p}\in\Lambda^+_{p,|M|}$ can be reached by successively adding fundamental weights $\vect{\omega}_{j,p}$. Using the recurrence relations \eqref{rec-rel}, we can thus write $\Delta_p(\vect{\mu})$ as a (finite) product of $\Delta_p(\vect{0})=1$
and ratios of the form $V_{\vect{\omega}_{j,p}}(\vect{\rho}_p+\sgn(M)\vect{\mu}^\prime)/V_{\vect{\omega}_{j,p}}(-\vect{\rho}_p-\sgn(M)(\vect{\mu}^\prime+\vect{\omega}_{j,p}))$, with $\vect{\mu}^\prime,\vect{\mu}^\prime+\vect{\omega}_{j,p}\in\Lambda^+_{p,|M|}$. Since Lemma \ref{Lemma:Vnu} ensures that these ratios are regular and positive, we have the following corollary.

\begin{corollary}\label{Cor:pos}
As long as $p\perp n$ and $\alpha,g>0$ satisfy Conditions 1--2 in Lemma \ref{Lemma:Vnu}, the lattice function $\Delta_p$ \eqref{Delta} is regular and positive, that is
$$
\Delta_p(\vect{\mu})>0\quad\text{at all}\quad\vect{\mu}\in\Lambda_{p,|M|}^+.
$$
\end{corollary}

Moreover, if we use the recurrence relations \eqref{rec-rel} to extend the definition of $\Delta_p(\vect{\mu})$ to lattice points $\vect{\mu}$ `one step' outside $\Lambda_{p,|M|}^+$, then \eqref{Vzer} immediately implies that the resulting lattice function vanishes at all boundary points, i.e.~$\vect{\mu}\notin\Lambda^+_{p,|M|}$ and $\vect{\mu}-\vect{\nu}\in\Lambda^+_{p,|M|}$ for some $\vect{\nu}\in S_n(\vect{\omega}_r)$ with $r=1,\dots,n-1$. To see that the extension is well-defined, note that the proof of Lemma \ref{Lemma:Vnu} entails that $V_{\vect{\nu}}(-\vect{\rho}_p-\sgn(M)(\vect{\mu}+\vect{\nu}))=0$ only if $\vect{\mu}\notin\Lambda^+_{p,|M|}$ and $\vect{\mu}+\vect{\nu}\in\Lambda^+_{p,|M|}$.

Having established the necessary properties of $\Delta_p$, we define the lattice function $\Psi_{\vect{0},p}\colon\vect{\rho}_p+\sgn(M)\Lambda^+_{p,|M|}\to\C$ by setting
\begin{equation}\label{Psi_0}
\Psi_{\vect{0},p}(\vect{\rho}_p+\sgn(M)\vect{\mu})=\frac{1}{\cN_{\vect{0}}^{1/2}}\Delta_p(\vect{\mu})^{1/2},\quad \vect{\mu}\in \Lambda_{p,|M|}^+,
\end{equation}
with the normalisation constant $\cN_{\vect{0}}$ determined by
$$
(\Psi_{\vect{0},p},\Psi_{\vect{0},p})_{p,M}=\frac{1}{\cN_{\vect{0}}}\sum_{\vect{\mu}\in\Lambda^+_{p,|M|}}\Delta_p(\vect{\mu})=1.
$$
In the $p=1, M>0$ case, van Diejen and Vinet \cite[Proposition 4.4]{vDV} arrived at a compact product formula for $\cN_{\vect{0}}$ by deriving a truncated version of a summation formula due to Aomoto, Ito and Macdonald \cite{Aom,Ito,Mac}. As we now demonstrate, the pertinent product formula holds true, after a few minor modifications, also in the general case.

\begin{proposition}
Assuming $p\perp n$ and $\alpha,g>0$ satisfy Conditions 1--2 in Lemma \ref{Lemma:Vnu} , we have
\begin{equation}
\cN_{\vect{0}}=2^{(n-1)(|M|-1)}n\prod_{k=1}^{n-1}\big(1+\sgn(M)kg:\sin_\alpha\big)_{|M|-1}.
\label{cN0val}
\end{equation}
\end{proposition}

\begin{proof}
Let us first consider the $M>0$ case. Following the discussion in Appendix A of \cite{vDV}, we find that Eq.~(A.4) remains valid upon setting $N=n-1$ and taking $\alpha\to p\alpha$, $A_{n-1}^+\to A_{n-1,p}^+$, $\Lambda_M^+\to\Lambda_{p,M}^+$ and (due to the $2\pi/\alpha$-antiperiodicity of $\sin\frac{\alpha}{2}x$) $\vect{\rho}\to\vect{\rho}_p$. This yields
\begin{multline}\label{sumForm}
\sum_{\vect{\mu}\in\Lambda_{p,M}^+}\prod_{\vect{a}\in A_{n-1,p}^+}\frac{\sin\frac{\alpha}{2}\langle \vect{a},\vect{\rho}_p+\vect{\mu}\rangle}{\sin\frac{\alpha}{2}\langle \vect{a},\vect{\rho}_p\rangle} \frac{\big(g+\langle \vect{a},\vect{\rho}_p\rangle:\sin_\alpha\big)_{\langle \vect{a},\vect{\mu}\rangle}}{\big(1-g+\langle \vect{a},\vect{\rho}_p\rangle:\sin_\alpha\big)_{\langle\vect{a},\vect{\mu}\rangle}}\\
=2^{(n-1)(M-1)}n\prod_{k=1}^{n-1}\big(1+kg:\sin_\alpha\big)_{M-1}
\end{multline}
for
$$
\alpha=\frac{2\pi p}{ng+M},\quad M\in\N,\quad g>0.
$$

By analytic continuation in $g$, \eqref{sumForm} extends to generic complex $g$. (We need only ensure that all denominators in the left-hand side are non-zero.) Taking $g\to -g$, which entails
$$
\vect{\rho}_p\to -\vect{\rho}_p\pmod{\tfrac{2\pi}{\alpha}\Lambda},\quad \alpha\to -\hat{\alpha},\ \ \hat{\alpha}\equiv \frac{2\pi p}{ng-M},
$$
the left-hand side of \eqref{sumForm} becomes
$$
\sum_{\vect{\mu}\in\Lambda_{p,M}^+}\prod_{\vect{a}\in A_{n-1,p}^+}\frac{\sin\frac{\hat{\alpha}}{2}\langle \vect{a},\vect{\rho}_p-\vect{\mu}\rangle}{\sin\frac{\hat{\alpha}}{2}\langle \vect{a},\vect{\rho}_p\rangle} \frac{\big(-g-\langle \vect{a},\vect{\rho}_p\rangle:\sin_{-\hat{\alpha}}\big)_{\langle \vect{a},\vect{\mu}\rangle}}{\big(1+g-\langle \vect{a},\vect{\rho}_p\rangle:\sin_{-\hat{\alpha}}\big)_{\langle\vect{a},\vect{\mu}\rangle}}.
$$
A direct computation reveals
$$
\big(\gamma-\langle \vect{a},\vect{\rho}_p\rangle:\sin_{-\hat{\alpha}}\big)_{\langle\vect{a},\vect{\mu}\rangle}=1\Big/\big(1-\gamma+\langle \vect{a},\vect{\rho}_p\rangle:\sin_{\hat{\alpha}}\big)_{-\langle\vect{a},\vect{\mu}\rangle},\quad \gamma\in\C.
$$
Moreover, using $ng=2\pi p/\hat{\alpha}+M$, it is straightforward to verify
$$
\prod_{k=1}^{n-1}\big(1-kg:\sin_{-\hat{\alpha}}\big)_{M-1}=\prod_{k=1}^{n-1}\big(1-kg:\sin_{\hat{\alpha}}\big)_{M-1}.
$$
We thus infer
\begin{multline*}
\sum_{\vect{\mu}\in\Lambda_{p,M}^+}\prod_{\vect{a}\in A_{n-1,p}^+}\frac{\sin\frac{\hat{\alpha}}{2}\langle \vect{a},\vect{\rho}_p-\vect{\mu}\rangle}{\sin\frac{\hat{\alpha}}{2}\langle \vect{a},\vect{\rho}_p\rangle} \frac{\big(-g+\langle \vect{a},\vect{\rho}_p\rangle:\sin_{\hat{\alpha}}\big)_{-\langle \vect{a},\vect{\mu}\rangle}}{\big(1+g+\langle \vect{a},\vect{\rho}_p\rangle:\sin_{\hat{\alpha}}\big)_{-\langle \vect{a},\vect{\mu}\rangle}}\\
=2^{(n-1)(M-1)}n\prod_{k=1}^{n-1}\big(1-kg:\sin_{\hat{\alpha}}\big)_{M-1},
\end{multline*}
which amounts to the $M<0$ case of \eqref{cN0val}.
\end{proof}

Due to Corollary \ref{Cor:pos}, we know that $\Psi_{\vect{0},p}$ is both regular and positive. To show that it has the desired eigenfunction property, we introduce the difference operators
$$
\hat{\cD}_{r,M}\equiv \sum_{\vect{\nu}\in S_n(\vect{\omega}_r)}V_{\vect{\nu}}(\vect{x})\exp(\sgn(M)\langle\vect{\nu},\partial/\partial \vect{x}\rangle),\quad r=1,\dots,n-1.
$$
The vanishing property \eqref{Vzer} of $V_{\vect{\nu}}(\vect{x})$ ensures that $\hat{\cD}_{r,M}$ admits a well-defined restriction onto $L^2(\vect{\rho}_p+\sgn(M)\Lambda^+_{p,|M|})$, given explicitly by
$$
\big(\hat{\cD}_{r,M}\phi\big)(\vect{\rho}_p+\sgn(M)\vect{\mu}) = \sum_{\substack{\vect{\nu}\in S_n(\vect{\omega}_r)\\ \vect{\mu}+\vect{\nu}\in\Lambda^+_{p,|M|}}} V_{\vect{\nu}}(\vect{\rho}_p+\sgn(M)\vect{\mu})\phi(\vect{\rho}_p+\sgn(M)(\vect{\mu}+\vect{\nu})).
$$
In this case, the $g$-values excluded in Condition 3 of Lemma \ref{Lemma:Vnu} cause no issues. Using the recurrence relations \eqref{rec-rel}, the next result is readily established by a direct computation.

\begin{corollary}\label{Cor:iTRel}
Let $p\perp n$ and $\alpha,g>0$ be such that Conditions 1--2 in Lemma \ref{Lemma:Vnu} are satisfied. Then, for any $\phi\in L^2(\vect{\rho}_p+\sgn(M)\Lambda^+_{p,|M|})$, we have
\begin{equation}\label{iTRel}
\big(\hat{\cS}_{r,M}(\Psi_{\vect{0},p}\phi)\big)(\vect{\rho}_p+\sgn(M)\vect{\mu}) = \big(\Psi_{\vect{0},p}\hat{\cD}_{r,M}(\phi)\big)(\vect{\rho}_p+\sgn(M)\vect{\mu}),
\end{equation}
with $r=1,\dots,n-1$.
\end{corollary}

At this point, we need the elementary symmetric functions
\begin{equation}\label{cEr}
\cE_r(\vect{u})\equiv \sum_{\vect{\nu}\in S_n(\vect{\omega}_r)}e^{\mathrm{i}\alpha\langle\vect{\nu},\vect{u}\rangle},\quad r=1,\dots,n-1
\end{equation}
and the weighted half-sum of positive roots $\vect{\rho}$ \eqref{rho}. For later reference, we note the identities
\begin{equation}\label{cEIds}
\cE_{n-r}(\vect{u}) = \cE_r(-\vect{u}),\quad \overline{\cE_r(\vect{u})} = \cE_r(-\vect{u}),\quad \vect{u}\in\R^n,
\end{equation}
where the former equality is a simple consequence of $\vect{\omega}_{n-r}\in S_n(-\vect{\omega}_r)$ and the latter is clear from \eqref{cEr}. Setting $\phi=1$ in \eqref{iTRel} and appealing to the Macdonald identity
$$
\sum_{\vect{\nu}\in S_n(\vect{\omega}_r)}V_{\vect{\nu}}(\vect{x}) = \cE_r(\vect{\rho})
$$
from Theorem 2.8 in \cite{Mac72}, we obtain the requisite joint eigenfunction property. For easy reference, the precise result is detailed in the following proposition.

\begin{proposition}\label{Prop:Psi0}
Under the assumptions of Corollary \ref{Cor:iTRel}, we have the joint eigenfunction property
$$
\hat\cS_{r,M}\Psi_{\vect{0},p} = \cE_r(\vect{\rho})\Psi_{\vect0,p},\quad r=1,\dots,n-1.
$$
\end{proposition}

For type (i) $g$-intervals \eqref{typei} such that $M>0$, the $\hat{\cD}_{r,M}$ coincide with the standard $A_{n-1}$ Macdonald-Ruijsenaars operators
\begin{equation}\label{cDr}
\hat{\cD}_r\equiv \sum_{\vect{\nu}\in S_n(\vect{\omega}_r)}V_{\vect{\nu}}(\vect{x})\exp(\langle\vect{\nu},\partial/\partial \vect{x}\rangle),\quad r=1,\dots,n-1.
\end{equation}
More generally, we have (cf.~\eqref{hcHrMtohcHr})
\begin{equation}\label{cDrMtocDr}
\hat{\cD}_{r,M}(g)=\begin{cases}
\hat{\cD}_r(g),& M>0\\
\hat{\cD}_{n-r}(-g),& M<0
\end{cases}
\end{equation}
which, when combined with Corollary \ref{Cor:iTRel}, suggests that a joint eigenbasis for the discrete difference operators $\hat{\cS}_{r,M}$ can be constructed in terms of $\Psi_{\vect{0},p}$ and the $A_{n-1}$ Macdonald polynomials. With this observation in mind, we devote the next section to a brief review of pertinent results from the theory of Macdonald polynomials, see e.g.~\cite{Mac95,Mac00,Sto11} for further details. A reader already familiar with this theory may wish to skip ahead to Subsection \ref{subsec:eigbasis}.

\subsection{$A_{n-1}$ Macdonald polynomials}\label{subsec:macdonald}
From Subsection \ref{subsec:rootnot}, we recall the notation \eqref{aj}--\eqref{QLamA} and the weighted half-sum of positive roots \eqref{rho}. The elementary symmetric functions \eqref{cEr} and the monomial symmetric functions
$$
m_{\vect{\lambda}}(\vect{x})\equiv \sum_{\vect{\mu}\in S_n(\vect{\lambda})}e^{\mathrm{i}\alpha\langle\vect{\mu},\vect{x}\rangle},\quad \vect{\lambda}\in\Lambda^+,
$$
will also be needed. Note that the dominance partial order on $\Lambda$, given by \eqref{dominance} with $p=1$, yields a partial ordering of the monomials $m_{\vect{\lambda}}(\vect{x})$.

The (monic) $A_{n-1}$ Macdonald polynomials can, in particular, be uniquely characterised as joint eigenfunctions of the difference operators $\hat{\cD}_r$ \eqref{cDr}. We now sketch the main steps in the proof of this important fact. First, the action of each $\hat{\cD}_r$ on the monomials $m_{\vect{\lambda}}(\vect{x})$ is shown to be lower-triangular:
$$
\hat{\cD}_r m_{\vect{\lambda}}(\vect{x})=c^r_{\vect{\lambda}\vect{\lambda}}m_{\vect{\lambda}}(\vect{x})+\sum_{\substack{\vect{\mu}\in\Lambda^+\\ \vect{\mu}\prec\vect{\lambda}}}c^r_{\vect{\lambda}\vect{\mu}}m_{\vect{\mu}}(\vect{x}),
$$
where the coefficients $c^r_{\vect{\lambda}\vect{\mu}}$ depend on $\alpha,g$ and
$$
c^r_{\vect{\lambda}\vect{\lambda}}=\cE_r(\vect{\rho}+\vect{\lambda}).
$$
Second, one verifies that, for generic (complex) parameter values, the diagonal elements $c^r_{\vect{\lambda}\vect{\lambda}}$ are pair-wise distinct:
$$
\cE_r(\vect{\rho}+\vect{\lambda})=\cE_r(\vect{\rho}+\vect{\mu})
\quad\Leftrightarrow\quad
\vect{\lambda}=\vect{\mu},\quad\vect{\lambda},\vect{\mu}\in\Lambda^+.
$$
(Here it suffices to consider a single difference operator $\hat{\cD}_r$.) Hence, since the difference operators $\hat{\cD}_r$ pair-wise commute, their actions on the monomials $m_{\vect{\lambda}}(\vect{x})$ can be simultaneously diagonalised, and the $A_{n-1}$ Macdonald polynomials can be defined as the unique (trigonometric) polynomials of the form
$$
P_{\vect{\lambda}}(\vect{x})=m_{\vect{\lambda}}(\vect{x})+\sum_{\substack{\vect{\mu}\in\Lambda^+\\ \vect{\mu}\prec\vect{\lambda}}}u_{\vect{\lambda}\vect{\mu}}m_{\vect{\mu}}(\vect{x}),\quad \vect{\lambda}\in\Lambda^+,
$$
that satisfy the system of difference equations
$$
\hat{\cD}_rP=\cE_r(\vect{\rho}+\vect{\lambda})P,\quad r=1,\dots,n-1.
$$

Note that the above proof is valid only for generic parameter values. Indeed, for special values of $\alpha,g$, the diagonal elements $c_{\vect{\lambda}\vect{\lambda}}^r$ may not be pair-wise distinct. Hence, as pointed out by van Diejen and Vinet \cite{vDV} (in the $p=1$ case), it is not entirely obvious that we can specialise the Macdonald polynomials $P_{\vect{\lambda}}$ to $\alpha,g$ along the hyperplanes \eqref{quant}. Restricting attention to $\vect{\lambda}\in\Lambda^+_{|M|}\equiv\Lambda^+_{1,|M|}$, which is sufficient for our purposes, the situation is readily clarified. Indeed, $\cE_r(\vect{u})=\cE_r(\vect{v})$ for all $r=1,\dots,n-1$ if and only if
$$
\vect{u}=\sigma(\vect{v}) \pmod{\tfrac{2\pi}{\alpha}Q}
$$
for some $\sigma\in S_n$. Since $\vect{\rho}+\Lambda^+_{|M|}\subset\overline{\Sigma}_{g,1}\subset\cA_n$ and each $S_n$-orbit in $\Lambda$ intersects $\Lambda^+$ only once, it follows that the $(n-1)$-tuples $(c^r_{\vect{\lambda}\vect{\lambda}})_{r=1}^{n-1}$ of diagonal elements are pair-wise distinct:
$$
\big(\cE_r(\vect{\rho}+\vect{\lambda})\big)_{r=1}^{n-1}=\big(\cE_r(\vect{\rho}+\vect{\mu})\big)_{r=1}^{n-1}
\quad\Leftrightarrow\quad
\vect{\lambda}=\vect{\mu},\quad \vect{\lambda},\vect{\mu}\in\Lambda^+_{|M|},
$$
which is enough to ensure that the actions of the difference operators $\hat{\cD}_r$ on the monomials $m_{\vect{\lambda}}$, $\vect{\lambda}\in\Lambda^+_{|M|}$, can be simultaneously diagonalised. The pertinent specialisation of the Macdonald polynomials $P_{\vect{\lambda}}$, $\vect{\lambda}\in\Lambda^+_{|M|}$, is therefore well-defined.

In what follows, we shall make use of the remarkable fact that the renormalised Macdonald polynomials
$$
\tilde{P}_{\vect{\lambda}}(\vect{x})=P_{\vect{\lambda}}(\vect{x})\prod_{\vect{a}\in A_{n-1}^+}\frac{(\langle\vect{a},\vect{\rho}\rangle:\sin_\alpha)_{\langle\vect{a},\vect{\lambda}\rangle}}{(\langle\vect{a},\vect{\rho}\rangle+g:\sin_\alpha)_{\langle\vect{a},\vect{\lambda}\rangle}},\quad\vect{\lambda}\in\Lambda^+,
$$
satisfy the self-duality relation
\begin{equation}\label{PselfDual}
\tilde{P}_{\vect{\lambda}}(\vect{\rho}+\vect{\mu})=\tilde{P}_{\vect{\mu}}(\vect{\rho}+\vect{\lambda}),\quad \vect{\lambda},\vect{\mu}\in\Lambda^+,
\end{equation}
which is due to Koornwinder \cite{Koo88} (see also Section VI.6 of \cite{Mac95}).

\subsection{An orthonormal eigenbasis}\label{subsec:eigbasis}
We proceed with the construction of an orthonormal eigenbasis of the discrete difference operators $\hat{\cS}_{1,M},\dots,\hat{\cS}_{n-1,M}$. In addition to the lattice function $\Delta_p$ \eqref{Delta}, the pertinent joint eigenfunctions $\Psi_{\vect{\lambda},p}$, $\vect{\lambda}\in\Lambda^+_{p,|M|}$, will be given in terms of a slight modification of the self-dual $A_{n-1}$ Macdonald polynomials $\tilde{P}_{\vect{\lambda}}$.

For $p\perp n$, we let $\sigma_p\in S_n$ be the unique element such that
$$
\sigma_p\big(\{\vect{\omega}_{1,p},\dots,\vect{\omega}_{n-1,p}\}\big)=\{\vect{\omega}_1,\dots,\vect{\omega}_{n-1}\},
$$
which maps $\Lambda^+_{p,|M|}$ one-to-one onto $\Lambda^+_{|M|}$. As noted in \eqref{VnuRefl}, the coefficient functions $V_{\vect{\nu}}(\vect{x})$ \eqref{Vnu} are invariant under the simultaneous reflections $(g,\vect{x})\to (-g,-\vect{x})$. In addition, we have the translation invariance
$$
V_{\vect{\nu}}(\vect{x}+(2\pi/\alpha)\vect{\mu})=V_{\vect{\nu}}(\vect{x}),\quad \vect{\mu}\in\Lambda.
$$
Thus, substituting $g\to \sgn(M)g$ and $\vect{x}\to \sgn(M)[\vect{x}+(2\pi/\alpha)(\vect{\omega}_{n-p+1,p}+\dots+\vect{\omega}_{n-1,p})]$ in the Macdonald-Ruijsenaars operators $\hat{\cD}_r$ \eqref{cDr} and appealing to the identity \eqref{cDrMtocDr}, we find that the (trigonometric) polynomials
\begin{equation}\label{Pmod}
\tilde{P}_{\sigma_p(\vect{\lambda})}\big(\sgn(M)g;\sgn(M)[\vect{x}+(2\pi/\alpha)(\vect{\omega}_{n-p+1,p}+\dots+\vect{\omega}_{n-1,p})]\big),\quad \vect{\lambda}\in\Lambda^+_p
\end{equation}
satisfy the difference equations
\begin{equation}\label{cDrMEqs}
\hat{\cD}_{r,M}\tilde{P}=\cE_r\big(\vect{\rho}(\sgn(M)g)+\sigma_p(\vect{\lambda})\big)\tilde{P},\quad r=1,\dots,n-1.
\end{equation}

Introducing the notation
$$
\vect{\check{\rho}}_p(g)\equiv \vect{\rho}_p(g) + \frac{2\pi}{\alpha}(\vect{\omega}_{n-p+1,p}+\dots+\vect{\omega}_{n-1,p}) = g\sum_{j=1}^{n-1}\vect{\omega}_{j,p},
$$
we define the lattice functions $\Psi_{\vect{\lambda},p}\colon \vect{\rho}_p+\sgn(M)\Lambda^+_{p,|M|}\to \C$ by letting
\begin{multline}\label{Psilamp}
\Psi_{\vect{\lambda},p}(g;\vect{\rho}_p+\sgn(M)\vect{\mu})\\
=\frac{1}{\cN_{\vect{0}}^{1/2}}\Delta_p(g;\vect{\lambda})^{1/2}\Delta_p(g;\vect{\mu})^{1/2} \tilde{P}_{\sigma_p(\vect{\lambda})}\big(\sgn(M)g;\vect{\check{\rho}}_p(\sgn(M)g)+\vect{\mu}\big),
\end{multline}
where $\vect{\lambda},\vect{\mu}\in\Lambda^+_{p,|M|}$. Note that this definition is consistent with the definition of $\Psi_{\vect{0},p}$ in \eqref{Psi_0}, since $\Delta_p(\vect{0})=1=\tilde{P}_{\vect{0}}$.

Observing
$$
\vect{\check{\rho}}_p(\sgn(M)g) = \sigma_p^{-1}\big(\vect{\rho}(\sgn(M)g)\big),
$$
the $S_n$-invariance and self-duality \eqref{PselfDual} of $\tilde{P}_{\vect{\lambda}}$ is readily seen to entail the following self-duality property.

\begin{proposition}
For $\vect{\lambda},\vect{\mu}\in\Lambda^+_{p,|M|}$, we have
\begin{equation}\label{PsiSelfDual}
\Psi_{\vect{\lambda},p}(\vect{\rho}_p+\sgn(M)\vect{\mu})=\Psi_{\vect{\mu},p}(\vect{\rho}_p+\sgn(M)\vect{\lambda}).
\end{equation}
\end{proposition}

The next result is easily deduced from the intertwining relations \eqref{iTRel} and the difference equations \eqref{cDrMEqs} (satisfied by the polynomials \eqref{Pmod}).

\begin{proposition}
Assuming that $p\perp n$ and $\alpha,g>0$ satisfy Conditions 1--2 of Lemma \ref{Lemma:Vnu}, we have, for all $\vect{\lambda}\in\Lambda^+_{p,|M|}$, the joint eigenfunction property
\begin{equation}\label{eigFunc}
\hat{\cS}_{r,M}\Psi_{\vect{\lambda},p}=\cE_r\big(\vect{\rho}(\sgn(M)g)+\sigma_p(\vect{\lambda})\big)\Psi_{\vect{\lambda},p},\quad r=1,\dots,n-1.
\end{equation}
\end{proposition}

Consequently, the lattice functions $\Psi_{\vect{\lambda},p}$ diagonalise the self-adjoint operators $\hat{\cH}_{r,M}$ \eqref{hHrM}. More precisely, appealing to the former identity in \eqref{cEIds}, we get
$$
\hat{\cH}_{r,M}\Psi_{\vect{\lambda},p}=E_r\big(\vect{\rho}(\sgn(M)g)+\sigma_p(\vect{\lambda})\big)\Psi_{\vect{\lambda},p},\quad r=1,\dots,n-1,
$$
with the eigenvalues given explicitly by
$$
E_r(\vect{u})\equiv \sum_{\vect{\nu}\in S_n(\vect{\omega}_r)}\cos\alpha\langle\vect{\nu},\vect{u}\rangle.
$$
Due to the self-duality relation \eqref{PsiSelfDual}, we have the very same result with respect to the spectral variable, namely
\begin{multline}\label{eigFunc2}
\sum_{\substack{\vect{\nu}\in S_n(\vect{\omega}_r)\\ \vect{\mu}+\vect{\nu}\in\Lambda^+_{p,|M|}}}W_{\vect{\nu}}(\vect{\rho}_p+\sgn(M)\vect{\mu})\Psi_{\vect{\mu}+\vect{\nu}}(\vect{\rho}_p+\sgn(M)\vect{\lambda})\\
+\sum_{\substack{\vect{\nu}\in S_n(\vect{\omega}_r)\\ \vect{\mu}-\vect{\nu}\in\Lambda^+_{p,|M|}}}W_{-\vect{\nu}}(\vect{\rho}_p+\sgn(M)\vect{\mu})\Psi_{\vect{\mu}-\vect{\nu}}(\vect{\rho}_p+\sgn(M)\vect{\lambda})\\
 = E_r\big(\vect{\rho}(\sgn(M)g)+\sigma_p(\vect{\lambda})\big)\Psi_{\vect{\mu}}(\vect{\rho}_p+\sgn(M)\vect{\lambda})
\end{multline}
for $r=1,\dots,n-1$. In the sense of Duistermaat and Gr\"unbaum \cite{DG86}, we have thus shown that $\Psi_{\vect{\lambda}}(\vect{\rho}_p+\sgn(M)\vect{\mu})$ satisfies a bispectral joint eigenvalue problem that is self-dual. This enables us to prove the following theorem.

\begin{theorem}
For $p\perp n$ and $\alpha,g>0$ satisfying Conditions 1--2 in Lemma \ref{Lemma:Vnu}, the lattice functions $\Psi_{\vect{\lambda},p}$ \eqref{Psilamp} form an orthonormal basis in $L^2(\vect{\rho}_p+\sgn(M)\Lambda^+_{p,|M|})$, that is
$$
(\Psi_{\vect{\lambda},p},\Psi_{\vect{\mu},p})_{p,M} = \delta_{\vect{\lambda}\vect{\mu}},\quad\vect{\lambda},\vect{\mu}\in\Lambda^+_{p,|M|},
$$
(where the Kronecker delta $\delta_{\vect{\lambda}\vect{\mu}}$ equals $1$ if $\vect{\lambda}=\vect{\mu}$ and $0$ otherwise).
\end{theorem}

\begin{proof}
Specialising \eqref{duality} to $\phi=\Psi_{\vect{\lambda},p}$ and $\psi=\Psi_{\vect{\mu},p}$, the eigenfunction property \eqref{eigFunc} and identities \eqref{cEIds} entail
$$
\big[\cE_r\big(\vect{\rho}(\sgn(M)g)+\sigma_p(\vect{\lambda})\big)-\cE_r\big(\vect{\rho}(\sgn(M)g)+\sigma_p(\vect{\mu})\big)\big]\big(\Psi_{\vect{\lambda},p},\Psi_{\vect{\mu},p}\big)_{p,M} = 0.
$$
Recall that $\cE_r(\vect{u})=\cE_r(\vect{v})$ for all $r=1,\dots,n-1$ if and only if $\vect{u}=\sigma(\vect{v}) \pmod{(2\pi/\alpha)Q}$ for some $\sigma\in S_n$. Since no two distinct points of $\vect{\rho}(\sgn(M)g)+\sigma_p(\Lambda^+_{p,|M|})$ satisfy such a relation, orthogonality of the lattice functions $\Psi_{\vect{\lambda},p}$ follows.

Introducing the function
$$
F_r(\vect{u})\equiv E_r\big(\sgn(M)[\vect{u}+(2\pi/\alpha)(\vect{\omega}_{n-p+1,p}+\dots+\vect{\omega}_{n-1,p})]\big),
$$
the $S_n$-invariance of $E_r(\vect{u})$ implies
$$
F_r\big(\vect{\rho}_p(g)+\sgn(M)\vect{\lambda}\big)=E_r\big(\vect{\rho}(\sgn(M)g)+\sigma_p(\vect{\lambda})\big).
$$
For $j=1,\ldots,n-1$ let $r=1,\ldots,n-1$ be such that $\vect{\omega}_{j,p}\in S_n(\vect{\omega}_r)$. Using the eigenfunction property \eqref{eigFunc2} and orthogonality of the $\Psi_{\vect{\lambda},p}$ to rewrite both sides of the equality $(F_r\Psi_{\vect{\mu},p},\Psi_{\vect{\mu}+\vect{\omega}_{j,p},p})_{p,M}=(\Psi_{\vect{\mu},p},F_r\Psi_{\vect{\mu}+\vect{\omega}_{j,p},p})_{p,M}$, we deduce the recurrence relation
$$
W_{\vect{\omega}_{j,p}}(\vect{\rho}_p+\sgn(M)\vect{\mu})N_{\vect{\mu}+\vect{\omega}_{j,p}} = W_{-\vect{\omega}_{j,p}}(\vect{\rho}_p+\sgn(M)(\vect{\mu}+\vect{\omega}_{j,p}))N_{\vect{\mu}}
$$
for the (squared) norms
$$
N_{\vect{\mu}}\equiv (\Psi_{\vect{\mu},p},\Psi_{\vect{\mu},p})_{p,M}.
$$
By \eqref{W}, the identity $V_{\vect{\nu}}(-\vect{x})=V_{-\vect{\nu}}(\vect{x})$ and Lemma \ref{Lemma:Vnu}, we have
$$
W_{\vect{\omega}_{j,p}}(\vect{\rho}_p+\sgn(M)\vect{\mu}) = W_{-\vect{\omega}_{j,p}}(\vect{\rho}_p+\sgn(M)(\vect{\mu}+\vect{\omega}_{j,p}))\neq 0,\quad \vect{\mu},\vect{\mu}+\vect{\omega}_{j,p}\in\Lambda^+_{p,|M|}.
$$
Since (by definition) $N_{\vect{0}}=1$ and each $\vect{\mu}\in\Lambda^+_{p,|M|}$ is of the form $\vect{\mu}=\sum_{j=1}^{n-1}m_j\vect{\omega}_{j,p}$ with $m_j\in\N_0$, the claimed orthonormality follows.
\end{proof}

\section{Discussion}\label{sec:discussion}
In this paper, we have quantised the $n$-particle compactified trigonometric RS models for all type (i) values \eqref{typei} of the coupling constant $g$, thus generalising earlier work by van Diejen and Vinet \cite{vDV} pertaining to $0<g<2\pi/\alpha n$. The appropriate Hilbert space for the quantum system turned out to be a finite-dimensional space of complex-valued functions supported on a uniform lattice over the $(n-1)$-simplex $\overline{\Sigma}_{g,p}$ \eqref{bSigma}, that is the (local) configuration space of the classical $n$-particle system. We realised the quantum Hamiltonians as pair-wise commuting and self-adjoint discrete difference operators acting in the space of such lattice functions, and obtained the corresponding normalised joint eigenfunctions in terms of discretised $A_{n-1}$ Macdonald polynomials with unitary parameters.

There are several related open problems as well as promising directions for future research. First of all, the task of quantising compactified trigonometric RS systems with type (ii) couplings remains. Most of the arguments used here will not work, as they hinge on the simple structure of the classical configuration space and often require that the parameter $p$ be coprime to $n$. To handle even the simplest non-trivial case ($n=4,p=2$) seems challenging, although we have made some progress and hope to return to this problem in the near future. We recall that the discrete orthogonality relations obtained by van Diejen and Vinet \cite{vDV} for Macdonald polynomials with unitary parameters, specialising to results due to Kirillov \cite{Kir} for particular $g$-values, were generalised to Macdonald polynomials with unitary parameters associated with (admissible pairs of) irreducible reduced crystallographic root systems by van Diejen and Emsiz \cite{vDE}, and analogous results were obtained by van Diejen and Stokman \cite{vDS} for multivariable $q$-Racah polynomials. This lends support to the natural expectation that our results can be generalised to compactified RS systems attached to root systems other than $A_{n-1}$. Finally, it would also be of interest to investigate the compactified elliptic models recently constructed in \cite{FG}.

\section*{Acknowledgements}
We would like to thank L\'{a}szl\'{o} Feh\'{e}r for suggesting the problem and Simon Ruijsenaars for helpful and inspiring discussions. We are also grateful to L\'{a}szl\'{o}, Simon, and the anonymous referees for their useful comments.

We are grateful to the London Mathematical Society for a Research in Pairs grant, which provided financial support for a visit by the former author to Loughborough University in October 2016, when many of the results presented in this paper were obtained. We would also like to thank the Erwin Schr\"{o}dinger Institute in Wien for hospitality and financial support during the workshop ``Elliptic Hypergeometric Functions in Combinatorics, Integrable Systems and Physics'' in March 2017, which enabled us to make further progress.

T.F.G.~was also supported in part by the Hungarian Scientific Research Fund (OTKA) under the grant K-111697 and the \'{U}NKP-16-3 New National Excellence Program of the Ministry of Human Capacities.


\begin{thebibliography}{99}

\bibitem{Aom}
K.~Aomoto,
\emph{On elliptic product formulas for Jackson integrals associated with reduced root systems},
J.~Algebr. Comb. \textbf{8} (1998) 115-126;
\doi{10.1023/A:1008629309210}

\bibitem{vDE}
J.F.~van Diejen and E.~Emsiz,
\emph{Orthogonality of Macdonald polynomials with unitary parameters},
Math.~Z.~{\bf 276} (2014) 517-542;
\doi{10.1007/s00209-013-1211-4}

\bibitem{vDS}
J.F.~van Diejen and J.V.~Stokman,
\emph{Multivariable $q$-Racah polynomials},
Duke Math.~J.~{\bf 91} (1998) 89-136;
\doi{10.1215/S0012-7094-98-09106-2}

\bibitem{vDV}
J.F.~van Diejen and L.~Vinet,
\emph{The quantum dynamics of the compactified trigonometric Ruijsenaars-Schneider model},
Commun.~Math.~Phys.~{\bf 197} (1998) 33-74;
\doi{10.1007/s002200050442}

\bibitem{vDV00}
J.F.~van Diejen and L.~Vinet (eds.),
Calogero-Moser-Sutherland Models,
Springer, New York, NY, 2000;
\doi{10.1007/978-1-4612-1206-5}

\bibitem{DG86}
J.J.~Duistermaat and F.A.~Gr\"{u}nbaum,
\emph{Differential equations in the spectral parameter},
Commun.~Math.~Phys.~\textbf{103} (1986) 177-240;
\doi{10.1007/BF01206937}

\bibitem{FG}
L.~Feh\'{e}r and T.F.~G\"{o}rbe,
\emph{Trigonometric and elliptic Ruijsenaars-Schneider systems on the complex projective space},
Lett.~Math.~Phys.~\textbf{106} (2016) 1429-1449;
\doi{10.1007/s11005-016-0877-z}

\bibitem{FK}
L.~Feh\'{e}r and C.~Klim\v{c}\'{i}k,
\emph{Spectra of the quantized action variables of the compactified Ruijsenaars-Schneider system},
Theor.~Math.~Phys. \textbf{171} (2012) 704-714;
\doi{10.1007/s11232-012-0068-8}

\bibitem{FKl}
L.~Feh\'{e}r and T.J.~Kluck,
\emph{New compact forms of the trigonometric Ruijsenaars-Schneider system},
Nucl.~Phys.~B \textbf{882} (2014) 97-127;
\doi{10.1016/j.nuclphysb.2014.02.020}

\bibitem{Ito}
M.~Ito,
\emph{On a theta product formula for Jackson integrals associated with root systems of rank two},
J. Math. Anal. Appl. \textbf{216} (1997) 122-163;
\doi{10.1006/jmaa.1997.5665}

\bibitem{Kir} A.A.~Kirillov Jr.,
\emph{On an inner product in modular tensor categories},
J.~Amer.~Math.~Soc.~\textbf{9} (1996) 1135-1169;
\doi{10.1090/S0894-0347-96-00210-X }

\bibitem{Koo88}
T.H.~Koornwinder,
\emph{Self-duality for q-ultraspherical polynomials associated with the root system $A_n$},
unpublished manuscript, 1988.

\bibitem{Mac72}
I.G.~Macdonald,
\emph{The Poincar\'{e} series of a Coxeter group},
Math.~Ann.~\textbf{199} (1972) 161-174;
\doi{10.1007/BF01431421}

\bibitem{Mac95}
I.G.~Macdonald,
Symmetric Functions and Hall Polynomials,
second edition, Oxford University Press, Oxford, 1995.

\bibitem{Mac00}
I.G.~Macdonald,
\emph{Orthogonal polynomials associated with root systems},
S\'em.~Lothar.~Combin.~\textbf{45}, Art.~B45a (2000) 40pp.
\arXiv{math/0011046}{[math.QA]}

\bibitem{Mac}
I.G.~Macdonald,
\emph{A formal identity for affine root systems},
pp.~195-211
in: Lie groups and Symmetric Spaces, ed.~S.G.~Gindikin,
Amer.~Math.~Soc.~Transl.~Ser.~2 {210}, 2003.
\doi{10.1090/trans2/210/14}

\bibitem{Ne99}
N.~Nekrasov,
\emph{Infinite-dimensional algebras, many-body systems and gauge theories},
in Moscow Seminar in Mathematical Physics (eds. A.~Morozov, M.A.~Olshanetsky), pp. 263-299,
AMS, Providence, RI, 1999;
\doi{10.1090/trans2/191/09}

\bibitem{RS}
S.N.M.~Ruijsenaars and H.~Schneider,
\emph{A new class of integrable systems and its relation to solitons},
Ann.~Physics \textbf{170} (1986), 370-405;
\doi{10.1016/0003-4916(86)90097-7}

\bibitem{Rui87}
S.N.M.~Ruijsenaars,
\emph{Complete integrability of relativistic Calogero-Moser systems and elliptic function identities},
Commun.~Math.~Phys.~\textbf{110} (1987) 191-213;
\doi{10.1007/BF01207363}

\bibitem{Rui90}
S.N.M.~Ruijsenaars,
\emph{Finite-dimensional soliton systems},
in Integrable and Superintegrable Systems (ed. B.~Kupershmidt), pp. 165-206,
Singapore, World Scientific, 1990;
\doi{10.1142/9789812797179_0008}

\bibitem{Rui}
S.N.M.~Ruijsenaars,
\emph{Action-angle maps and scattering theory for some finite-dimensional integrable systems III. Sutherland type systems and their duals},
Publ.~RIMS \textbf{31} (1995) 247-353;
\doi{10.2977/prims/1195164440}

\bibitem{Rui99}
S.N.M.~Ruijsenaars,
\emph{Systems of Calogero-Moser type},
in Particles and Fields (eds. G.~Semenoff, L.~Vinet), pp. 251-352,
Springer, New York, NY, 1999;
\doi{10.1007/978-1-4612-1410-6_7}

\bibitem{Sto11}
J.V.~Stokman,
\emph{Macdonald-Koornwinder polynomials},
preprint (2011);
\arXiv{1111.6112}{[math.QA]}

\end{thebibliography}
\end{document}